\documentclass[10pt,journal,compsoc]{IEEEtran}
\usepackage{array}
\usepackage{makecell}
\usepackage{textcomp}
\usepackage{stfloats}
\usepackage{url}
\usepackage{verbatim}
\usepackage{graphicx}
\usepackage[caption=false,font=footnotesize,labelfont=rm,textfont=rm]{subfig}
\usepackage{float}
\usepackage{amsmath,amssymb,amsfonts}
\usepackage{mathrsfs}
\usepackage{amsthm}

\newtheorem{theorem}{Theorem} 

\usepackage{epstopdf}
\usepackage[linesnumbered,algoruled,boxed,lined]{algorithm2e}
\usepackage[noend]{algpseudocode}
\usepackage[section]{placeins}
\allowdisplaybreaks[4]
\usepackage{multirow}
\usepackage{booktabs}
\usepackage{tabularx}
\usepackage{cleveref}
\usepackage{color}
\usepackage{bm}
\usepackage{bbm}

\graphicspath{ {images/} }
\def\BibTeX{{\rm B\kern-.05em{\sc i\kern-.025em b}\kern-.08em
    T\kern-.1667em\lower.7ex\hbox{E}\kern-.125emX}}
\usepackage{cite}

\usepackage[caption=false,font=footnotesize,labelfont=rm,textfont=rm]{subfig}
\usepackage{float} 
\usepackage{comment}
\usepackage{cleveref}
\usepackage{multirow}
\usepackage{bm}
\usepackage{tabularx}
\usepackage{cite}
\usepackage[caption=false,font=footnotesize,labelfont=rm,textfont=rm]{subfig}
\ifCLASSINFOpdf

\begin{document}
\title{Joint AoI and Handover Optimization in Space-Air-Ground Integrated Network}

\author{Zifan Lang, Guixia Liu, 
        Geng Sun,~\IEEEmembership{Senior Member,~IEEE,}
        Jiahui Li, 
        Jiacheng Wang,\\ 
        Weijie Yuan,~\IEEEmembership{Senior Member,~IEEE,}
        Dusit Niyato,~\IEEEmembership{Fellow,~IEEE,}
        Dong In Kim,~\IEEEmembership{Life Fellow,~IEEE}
\thanks{Zifan Lang, Guixia Liu, and Jiahui Li are with the College of Computer Science and Technology, Jilin University, Changchun 130012, China (e-mails: langzf23@mails.jlu.edu.cn; liugx@jlu.edu.cn; lijiahui@jlu.edu.cn).}   
\thanks{Geng Sun is with the College of Computer Science and Technology, Jilin University, Changchun 130012, China, and with Key Laboratory of Symbolic Computation and Knowledge Engineering of Ministry of Education, Jilin University, Changchun 130012, China; he is also affiliated with the College of Computing and Data Science, Nanyang Technological University, Singapore 639798 (e-mail: sungeng@jlu.edu.cn). }
\thanks{Jiacheng Wang and Dusit Niyato are with the College of Computing and Data Science, Nanyang Technological University, Singapore (e-mails: jiacheng.wang@ntu.edu.sg; dniyato@ntu.edu.sg).}
\thanks{Weijie Yuan is with the Department of Electronic and Electrical Engineering, Southern University of Science and Technology, Shenzhen 518055, China (e-mail: yuanwj@sustech.edu.cn).}
\thanks{Dong In Kim is with the Department of Electrical and Computer Engineering, Sungkyunkwan University, Suwon 16419, South Korea (e-mail:
dongin@skku.edu).}
\thanks{\textit{(Corresponding authors: Guixia Liu and Geng Sun.)}}
}

\IEEEtitleabstractindextext{
    \begin{abstract}
    Despite the widespread deployment of terrestrial networks, providing reliable communication services to remote areas and maintaining connectivity during emergencies remains challenging. Low Earth orbit (LEO) satellite constellations offer promising solutions with their global coverage capabilities and reduced latency, yet struggle with intermittent coverage and limited communication windows due to orbital dynamics. This paper introduces an age of information (AoI)-aware space-air-ground integrated network (SAGIN) architecture that leverages a high-altitude platform (HAP) as intelligent relay between the LEO satellites and ground terminals. Our three-layer design employs hybrid free-space optical (FSO) links for high-capacity satellite-to-HAP communication and reliable radio frequency (RF) links for HAP-to-ground transmission, and thus addressing the temporal discontinuity in LEO satellite coverage while serving diverse user priorities. Specifically, we formulate a joint optimization problem to simultaneously minimize the AoI and satellite handover frequency through optimal transmit power distribution and satellite selection decisions. This highly dynamic, non-convex problem with time-coupled constraints presents significant computational challenges for traditional approaches. To address these difficulties, we propose a novel diffusion model (DM)-enhanced dueling double deep Q-network with \underline{a}ction decomposition and \underline{s}tate transformer encoder (DD3QN-AS) algorithm that incorporates transformer-based temporal feature extraction and employs a DM-based latent prompt generative module to refine state-action representations through conditional denoising. Simulation results highlight the superior performance of the proposed approach compared with policy-based methods and some other deep reinforcement learning (DRL) benchmarks. Moreover, performance analysis under various system settings verifies the robustness of the proposed approach. 
%Extensive simulations demonstrate that our DD3QN-AS algorithm achieves approximately 15$\%$ lower AoI values while maintaining reduced handover frequency compared to state-of-the-art baselines. Furthermore, our results reveal that the latest deadline first (LDF) queuing policy delivers optimal performance balance. Moreover, performance analysis under different denoising steps of DD3QN-AS and different numbers of ground users verifies the robustness of the proposed approach.  
    \end{abstract}
    
}
\maketitle


\IEEEpeerreviewmaketitle

\section{Introduction}
\IEEEPARstart{A}{s} a promising candidate for network architecture capable of supporting access to massive Internet of Things (IoT) devices, satellite communication has garnered significant interest from both academia and industry owing to its wide communication coverage, long transmission distances, and flexible deployment capabilities~\cite{Wang2025}. In particular, low Earth orbit (LEO) satellite constellations offer reduced latency and improved link budgets compared to traditional geostationary satellites, making them increasingly attractive for time-sensitive applications and remote area coverage~\cite{Li2024}. These advantages position LEO satellite networks as critical infrastructure for global IoT connectivity, disaster relief, emergency communications, and various monitoring systems requiring widespread geographical coverage~\cite{fan2025satellite}. 

\par LEO satellites face inherent limitations in providing continuous service to specific geographical regions due to their orbital dynamics, resulting in intermittent coverage and limited communication windows with ground terminals~\cite{Hui2025}. In particular, the temporal discontinuity issue in LEO satellite coverage becomes particularly problematic when serving multiple ground users with diverse priority levels and data freshness requirements~\cite{Chen2024a}. Furthermore, the limited communication opportunities during satellite passes create resource contention among ground users. Such contention leads to inefficient spectrum utilization and increased transmission delays, thus increasing the age of information (AoI)\cite{Chen2025}. As such, the aforementioned challenges require an innovative architectural approach that can bridge coverage gaps while intelligently managing information dissemination to ground users based on their specific requirements and priorities. %

\par To this end, we consider to design a space-air-ground integrated network (SAGIN) architecture to enhance the information freshness~\cite{zhang2025aoi}. In this architecture, high-altitude platform (HAP) serves as a relay between the LEO satellites and ground users. Moreover, free-space optical (FSO) links enable high-capacity satellite-to-HAP communication, while robust radio-frequency (RF) links support reliable HAP-to-ground transmission~\cite{li2025llm}. As such, the quasi-stationary nature of HAP provides more predictable link budgets and simplified tracking requirements. % compared to direct LEO communications. 
By bridging space and ground layers, the hybrid SAGIN architecture enhances scalability for large-scale IoT deployments through efficient user access management and the adoption of AoI-aware transmission policies that prioritize timely data delivery. 

\par However, designing an effective and reliable HAP-assisted relay system involves several complexities. \textit{Firstly}, the dynamic topology created by LEO satellite movement necessitates adaptive resource allocation and link scheduling strategies to maintain continuous connectivity. \textit{Secondly}, the heterogeneous nature of FSO and RF links introduces complex interrelationships in resource allocation decisions. 
\textit{Finally}, the uncertainty and variability introduced by dynamic updates of satellite data demand robust optimization strategies that can adapt in real time. However, existing approaches in the literature, such as conventional convex optimization and static resource allocation schemes, struggle to simultaneously address the high dynamics, heterogeneity, and time-varying uncertainty inherent in this system, which makes them inadequate for our scenario. 

%Different from the previous studies
\par Accordingly, this paper considers an AoI-aware HAP-assisted relay system for SAGINs and proposes an enhanced deep reinforcement learning (DRL) algorithm to minimize the AoI and handover frequency of satellites. The major contributions of this work are outlined as follows: 
\begin{itemize}
    \item \textit{Dynamic Multi-tier SAGIN Architecture for AoI-aware Communication:} We design a three-layer HAP-assisted downlink architecture that leverages hybrid FSO/RF communication. The architecture enables relay-assisted AoI-aware packet transmission from LEO satellites to ground terminals and supports dynamic satellite selection. We tailer the system for environments requiring delay-sensitive services such as real-time monitoring and disaster response. 
    
    \item \textit{Dynamic Joint Optimization Problem with Complex Interdependencies:} We model the system to capture the complex interplay between the LEO satellites mobility and dynamic channel conditions across multi-layer networks. Specifically, we find significant interdependencies between satellite handover processes and resource allocation decisions at the HAP layer. To this end, we formulate a joint optimization problem that simultaneously balances the AoI and satellite handover frequency through strategic handover decisions and optimal transmit power distribution. Moreover, the orbital dynamics of LEO satellites, stochastic channel conditions, and mixed decision space lead to highly dynamic and non-convex characteristics with time-coupled constraints. Thus we require an enhanced solution approach with superior adaptability to temporal variations. %To this end, we formulate a joint optimization problem that simultaneously balances the AoI and satellite handover frequency through strategic handover decisions and optimal transmit power distribution.
    \item \textit{Diffusion Model (DM)-enhanced DRL-based Algorithm:} Traditional optimization approaches, such as conventional convex optimization and static resource allocation schemes, struggle with the dynamic and non-convex nature of the optimization problem. As such, we propose a DM-enhanced dueling double deep Q-network with \underline{a}ction decomposition and \underline{s}tate transformer encoder (DD3QN-AS) algorithm to solve the optimization problem. %that overcomes the limitations of standard reinforcement learning approaches in highly dynamic environments. 
    Specifically, DD3QN-AS incorporates the state transformer encoder to achieve more expressive value estimation across temporal variations. Moreover, the algorithm employs the DM to refine latent state-action representations, thereby achieving fast and adaptive convergence and enhanced stability performance. 
    
    \item \textit{Extensive Performance Evaluation:} Simulation results demonstrate that the proposed DD3QN-AS algorithm achieves faster convergence compared to various baseline algorithms. Moreover, DD3QN-AS achieves the lowest AoI values while maintaining a reduced satellite handover frequency. In addition, we evaluate the impact of different HAP queue scheduling policies on the DD3QN-AS algorithm, and reveal that the latest deadline first policy achieves the best performance. To further verify the robustness, we conduct the performance analysis of the proposed DD3QN-AS under different denoising steps. %In addition, we evaluate the impact of different HAP queue scheduling policies on the DD3QN-AS algorithm, and reveal that the latest deadline first policy achieves the best balance between AoI minimization and handover stability.
\end{itemize}
\par The rest of this paper is arranged as follows. Section~\ref{sec:related work} reviews the related works. Section~\ref{sec:system model} presents the models and problem. Section~\ref{sec:DRL} proposes the proposed algorithm. Section~\ref{sec:simulation} shows the simulation results, and Section~\ref{sec:conclusion} concludes the paper.

%
%section
%Related Work
\section{Related Work}
\label{sec:related work}

\par In this work, we aim to propose an HAP-assisted AoI-aware downlink system for SAGIN by optimizing the transmit power distribution and satellite handover decisions. In the following, we review several key related works to highlight the novelty of this research. %Furthermore, Table \ref{tab:comparison} summarizes the main distinctions between these works and this work. 

\subsection{FSO-based Satellite Networks Architectures}

%\par Due to its exceptional bandwidth capabilities and minimal spectrum regulation constraints, FSO communication has emerged as a valuable technique for next-generation satellite communications with potential to overcome traditional RF limitations. 
\par Due to its exceptional bandwidth capabilities and minimal spectrum regulation constraints, FSO communication has emerged as a promising technique for next-generation systems to address the inherent limitations of conventional RF technologies.
For instance, the authors in~\cite{Le2023} proposed a cross-layer framework for FSO burst transmissions in SAGIN. 
The authors in~\cite{Mao2024} introduced an intelligent hierarchical routing strategy for ultra-dense FSO-based LEO satellite networks to satisfy diverse quality of service requirements of terrestrial applications. 
As illustrated in~\cite{Dabiri2025}, the authors developed a modulating retro-reflector-based optical downlink system for resource-constrained CubeSats. % that enables simultaneous tracking and high-speed communication
In addition, the authors in~\cite{chen2025free} developed an FSO semantic communication system for satellite remote sensing image transmission. 
However, these FSO-based architectures face significant challenges from atmospheric turbulence, cloud coverage, and weather conditions, which can severely degrade link quality and reliability, particularly for satellite-to-ground communications. 

\par To overcome these technology limitations, hybrid FSO/RF transmission architectures have been extensively investigated to offer a more robust solution for practical deployment scenarios in space-air-ground integrated networks. Specifically, the authors in~\cite{Ma2022} investigated the secrecy outage performance of a hybrid FSO/RF cooperative SAGIN. % where a terrestrial source communicates with a satellite through a cache-enabled aerial relay while being potentially intercepted by aerial eavesdroppers
The authors in~\cite{xu2024cooperative} proposed a cooperative FSO/RF SAGIN with an adaptive scheme that switches between FSO-only and combined modes based on link quality. 
Likewise, the authors in~\cite{mashiko2025combined} proposed a combined control method for optimizing HAPs placement and coverage area allocation in hybrid FSO/RF SAGINs. 
Moreover, the authors in~\cite{fan2025gato} addressed a global transmission optimization strategy for SAGIN-assisted IoT data collection to maximize data upload capacity in 6G networks. 
However, these studies mainly focus on establishing static communication links while neglecting various dynamic elements, particularly inter-satellite handover complexities. %However, these works predominantly concentrate on static communication link establishment while overlooking several dynamic aspects, such as satellite mobility patterns, inter-satellite handover complexities.

\subsection{Optimization Metrics in SAGINs}

\par Optimizing performance metrics while maintaining service continuity represents a primary challenge in SAGINs. 
For example, the authors in~\cite{Gao2024} minimized the time-averaged network cost in SAGIN through joint optimization of task hosting, computation offloading and resource allocation strategies. 
The authors in~\cite{tan2024outage} introduced an analytical framework to optimize unmanned aerial vehicles (UAV) positioning in SAGIN and improve transmission fairness. 
As illustrated in~\cite{zhao2025energy}, the authors optimized the total energy consumption in probabilistic semantic communication for SAGIN by optimizing the semantic compression ratios and computational overhead. %communication and computation 
Likewise, the authors in~\cite{zhang2025integrated} minimized weighted energy consumption in SAGIN by joint optimization of beamforming, task offloading ratios, and computation resource allocation. 
However, the abovementioned works focus on optimizing aggregated static link reliability metrics and thus neglect the temporal dynamics and freshness requirements of data transmission in highly dynamic SAGIN environments. 

\par Furthermore, the concept of AoI has gained significant attention as a performance metric for time-sensitive applications in SAGINs. For example, the authors in~\cite{Ke2025} evaluated information freshness in multi-hop satellite IoT systems by deriving closed-form AoI expressions for various automatic repeat request (ARQ) and hybrid ARQ schemes. 
Similarly, the authors in~\cite{Deng2024} minimized the AoI and energy cost in SAGIN by introducing a customizable update cost metric for polar-coded systems. 
Moreover, the the authors in~\cite{Wang2024} proposed an AoI-thresholding transmission scheduling mechanism to optimize the weighted sum of control cost and power consumption in SAGIN. In addition, the authors in~\cite{zhang2025aoi} presented a two-stage optimization framework for SAGIN-based IoT information collection that minimized AoI and optimized HAP and UAV coverage areas and heights. 

\par Despite the advancements, these research efforts largely overlook the impact of satellite handover mechanisms, which fundamentally alter the network topology, disrupt information freshness guarantees, and critically influence overall system reliability in operational LEO satellite constellations. 

\subsection{Optimization Methods in SAGINs}
\par Various schemes are used to handle the complex optimization problems in the SAGIN system. 
For instance, the authors in~\cite{Nguyen2024} developed an alternating optimization scheme for computation offloading in hybrid edge–cloud based SAGIN, while integrating successive convex approximation to resolve non-convexity. 
The authors in~\cite{Bao2024} optimized SAGIN delivery rates via a decentralized caching status-aware scheme. %that integrated network-layer placement with physical-layer delivery to exploit coded multicast opportunities 
Likewise, the authors in~\cite{Sun2024} addressed post-disaster SAGIN optimization through a joint task offloading and resource allocation framework, which combines game theory, convex optimization, and evolutionary computation. 
Moreover, the authors in~\cite{Mao2025} employed extended Kalman filter-based user tracking and Bernstein-type inequality reformulation to optimize outage-constrained energy efficiency in SAGIN. 
However, these methods still exhibit limitations in dynamic adjustment capabilities and lack timely feedback mechanisms. 

\par To address these limitations, recent academic efforts have started embedding learning-based approaches for optimizing SAGIN implementations. Specifically, the authors in~\cite{Elmahallawy2024} proposed a communication-efficient federated learning system for LEO satellite constellations, and designed a communication topology with aggregation strategies to balance models across orbits and mitigate Doppler effects. 
Analogously, the authors in~\cite{Qin2024} developed a differentiated federated reinforcement learning approach for SAGIN traffic offloading. The authors in~\cite{Chen2024} optimized the quality of service and resource utilization in SAGIN by using self-attention-based traffic prediction and knowledge-distilled reinforcement learning. Moreover, the authors in~\cite{Sun2024a} proposed a two-stage decoupled optimization approach that addressed task scheduling in competitive SAGIN environments by first using an auction algorithm, followed by a distributed proximal policy optimization reinforcement learning algorithm. Nevertheless, these methods often suffer from slow convergence and instability when applied to high-dimensional state spaces characteristic of satellite networks with numerous users and complex orbital mechanics. 

\subsection{Summary}
\par Despite substantial progress made by existing research in SAGINs, key challenges remain in managing AoI in dynamic HAP-assisted environments. \textit{Firstly}, although hybrid FSO/RF architectures enhance spectral efficiency and link robustness, they primarily target static link establishment, and overlook challenges posed by frequent satellite handovers and HAP relay coordination in dynamic environments. 
\textit{Secondly}, a broad range of SAGIN optimization studies aggregate static link reliability metrics, while ignoring the time-coupled evolution of AoI, especially when frequent satellite handovers reshape the service topology.  \textit{Finally}, the advanced optimization and learning frameworks have begun addressing the scale and heterogeneity of SAGIN. However, existing DRL methods often struggle with high-dimensional temporal state representations, multi-objective tradeoffs, and training instability in rapidly varying and user-diverse environments. Accordingly, we aim to devise an online algorithm with enhanced adaptability to jointly optimize AoI and satellite handover decisions under dynamic SAGIN conditions. 

\section{System Model}\label{sec:system model}
\begin{table}[t]
\centering
\caption{Main notations}
\label{tab:notations}
%\begin{tabularx}{3.5in}{p{2cm}p{6.5cm}}
\begin{tabular}{ll}
\toprule
\textbf{Symbol} & \textbf{Definition} \\ \midrule
$B_j$ & Bandwidth allocated to user $j$ \\
$\mathcal{H}$ & High-altitude platform (HAP) \\
$h_{FSO}$ &Channel gain between satellite and HAP \\
$h_j$ & Channel gain between HAP and user $j$ \\
$i$, $N_{S}$, $\mathcal{N}$ & The index, the number, and the set of LEO satellites \\
$j$, $N_{U}$, $\mathcal{U}$ & The index, the number, and the set of ground users \\ 
$l_t$ & Index of the satellite selected by the HAP at time slot $t$ \\
$L_q$ & Maximum length of HAP buffer queue \\
$N_t$ & Number of satellite handovers from time slot 0 up to $t$ \\
$P_S$ & Transmit power of the satellite \\
$P_j(t)$ & Transmit power allocated to user $j$ at time slot $t$ \\
$P_{HAP}$ & Maximum total transmit power of the HAP \\
$q_t$ & HAP queue at time slot $t$ \\
$R_{FSO}(t)$ & Achievable rate between satellite and HAP $j$ at time slot $t$\\
$R_j(t)$ & Achievable rate between HAP and user $j$ at time slot $t$\\
$t$, $T$, $\mathcal{T}$ & The index, the number, and the set of time slots \\
$\gamma_{HAP}$ & SNR of the satellite-to-HAP FSO link \\
$\gamma_j$ & SNR of the HAP-to-user RF link for user $j$ \\
$\Gamma_i[t]$ & Indicator of new packet generation at satellite $i$ \\
$\theta_i[t]$ & AoI of satellite $i$ at the satellite side at time $t$ \\
$\delta_i[t]$ & AoI of satellite $i$ at the HAP at time $t$ \\
$\Delta_{i,j}[t]$ & AoI of satellite $i$ at user $j$ at time $t$ \\
$\Lambda_{th}^S, \Lambda_{th}^H$ & SNR thresholds for FSO and RF links \\
%$T_i, T_j$ & Transmission delays of satellite-to-HAP and HAP-to-user \\
\bottomrule
\end{tabular}
\end{table}
\begin{comment}
    $\iota, \Omega, \omega, \epsilon, \rho, \nu$ & Orbital parameters: inclination angle, right ascension, \\
 &argument of perigee, eccentricity, semi-major axis, \\
 &true anomaly \\
\end{comment}
\par In this section, we first present the system overview. Then, we detail the considered models, including the satellite handover, communication, HAP queue and AoI models, to characterize the optimization objectives and decision variables.  The main notations are presented in Table~\ref{tab:notations}. 
%an HAP-Assisted AoI-Aware Downlink System in SAGINs

\subsection{System Overview}
% 根据系统模型图静态描述系统的构成，第二段描述动态（系统是怎么运行的），第三段介绍3D坐标符号。
\par As shown in Fig.~\ref{fig:system model}, we consider an HAP-assisted AoI-aware downlink system for SAGIN. In this system, the LEO satellites denoted as $\mathcal{S}\triangleq \{i\mid 1,2,\dots,N_{S}\}$ are equipped with high-resolution remote sensing devices to capture real-time data of the Earth surface. Due to the high orbital velocity of LEO satellites, direct transmission of large-scale data to ground terminals is challenging, as the limited coverage time and rapid movement restrict stable communication links~\cite{Xiao2024}. In this case, a single HAP, denoted $\mathcal{H}$, is deployed as a relay between the LEO satellites and the ground users, where the satellites communicate with the HAP via FSO links. We consider these ground users denoted as $\mathcal{U}\triangleq \{j\mid 1,2,\dots,N_{U}\}$ are randomly distributed and assigned to the HAP coverage area. 
%Due to long distances and atmospheric propagation, satellite-to-ground communication may suffer high latency, which leads to high AoI of the data. In this case, a single high-altitude platform (HAP), denoted $\mathcal{H}$, is deployed as a relay between the LEO satellites and the ground users. We consider these ground users denoted as $\mathcal{U}\triangleq \{j\mid 1,2,\dots,N_{U}\}$ are randomly distributed assigned to the HAP coverage area. 
%先采用固定用户，后期根据需求可调整为动态用户
%卫星生成数据方式待完善，比如概率

\par Without loss of generality, we consider a discrete-time system over the timeline $\mathcal{T} \triangleq \{t \mid 1, 2, \dots, T\}$. At each time slot, a subset of LEO satellites with sufficient spectrum resources and favorable angles can communicate with the HAP. Furthermore, there is no direct communication between the LEO satellites, while the HAP can access only one LEO satellite at a time, and the selected satellite is indexed by $l_t$. Data from the selected satellite are transmitted to the HAP through an FSO link and temporarily stored in $N_U$ virtual queues at the HAP. However, the HAP has a fixed buffer size, and once the buffer reaches capacity, it cannot receive additional data from the satellite. This overflow leads to data loss and increases the AoI at the satellite. In this case, a queue scheduling mechanism is employed to prioritize and efficiently transmit the buffered data to ground users. 

\begin{figure}[t]
    \centering
    \includegraphics[width=0.9\linewidth]{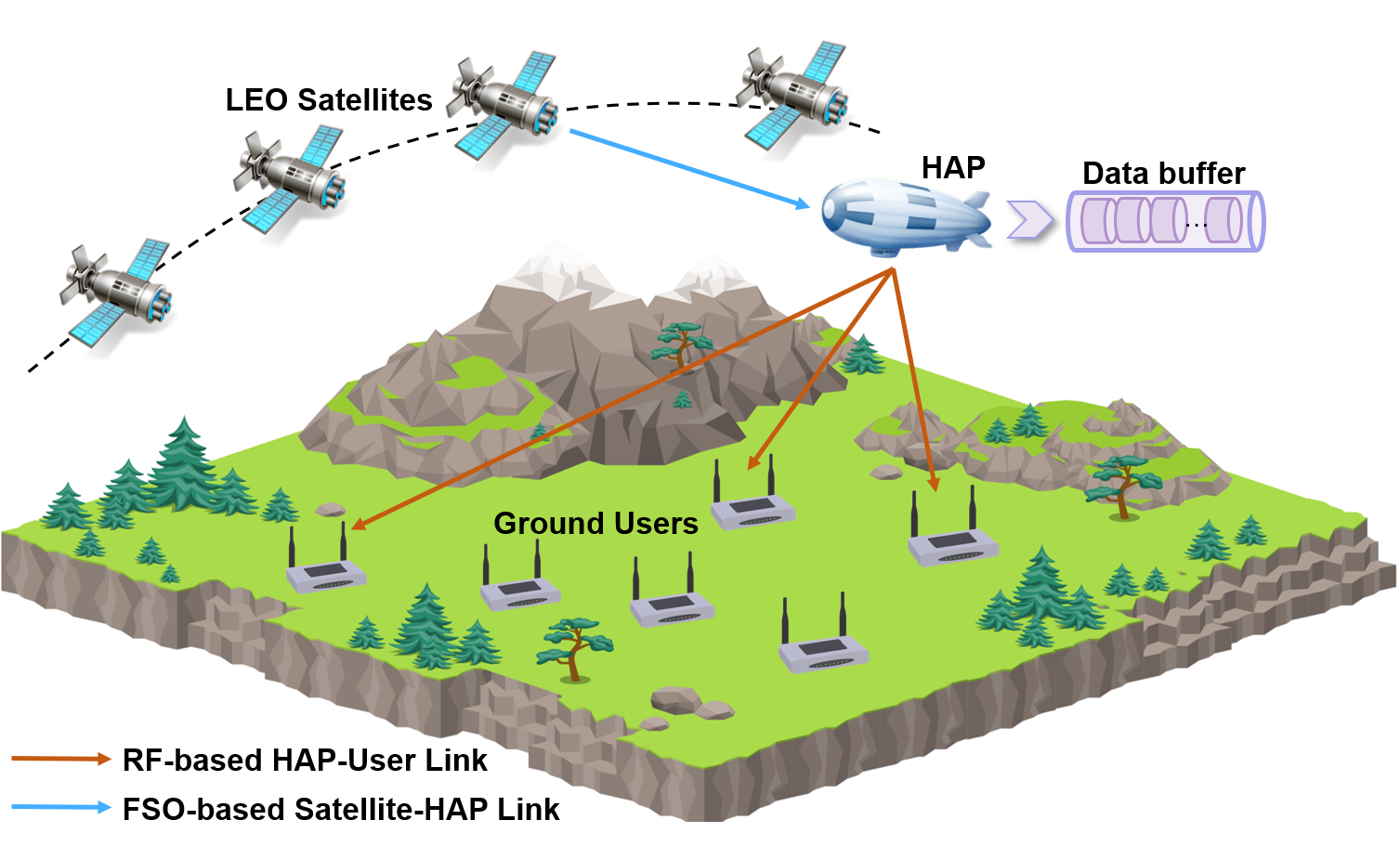}
    \caption{An HAP-assisted AoI-aware downlink system for SAGINs.}
    \label{fig:system model}
\end{figure}
\par We also consider a Cartesian coordinate system, where the locations of the HAP, $i$-th LEO satellite, and $j$-th user at time slot $t$ are represented as $\boldsymbol{c}^{HAP}=[x^{HAP},y^{HAP},z^{HAP}]$, $\boldsymbol{c}_{i,t}^{S}=[x^{S}_{i,t},y^{S}_{i,t}, z^{S}_{i,t}]$, $\boldsymbol{c}_{j,t}^{U}=[x^j_t,y^j_t, z^j_t]$, respectively.

\par As such, the limited spectral resources of satellites contribute to the uncertainty of availability. In the following, we introduce the satellite handover model and the satellite-HAP-user communication models, and then present the HAP queue model and AoI model. 

\subsection{Satellite Handover Model}
\par LEO satellites refer to satellites orbiting Earth at low altitudes, typically between 160 and 2000 kilometers. This rapid orbital motion results in a constantly changing satellite constellation, where each satellite is visible to the HAP for a limited time before moving out of range. Moreover, the orbital parameters of LEO satellites can be mathematically described by the tuple \(<\iota, \Omega, \omega, \epsilon, \rho, \nu>\)~\cite{Li2024}, \textit{i.e.}, 
\begin{itemize}
  \item \textit{Inclination Angle ($\iota$):} The angle between the orbital plane and the equator. 
  \item \textit{Right Ascension of Ascending Node ($\Omega$):} The angle between the vernal equinox and the point where the satellite crosses the equatorial plane.
  \item \textit{Argument of the Perigee ($\omega$):} The angle between the ascending node and the perigee, the point closest to Earth.
  \item \textit{Eccentricity ($\varepsilon$):} Describes the shape of the orbit. 
  \item \textit{Semi-Major Axis ($\varrho$):} Half of the longest diameter of the elliptical orbit.
  \item \textit{True Anomaly ($\nu$):} The angle between the perigee direction and the satellite's current position.
\end{itemize}

\par For simplicity, LEO satellites are typically modeled with circular orbits (\( \epsilon = 0 \)), where the orbital radius is given by \( H_i = h_i + R_e \), with \( h_i \) representing the satellite altitude and \( R_e \) the radius of Earth. The orbital period \( \tau_i \) is calculated as $\tau_i = 2\pi \sqrt{{H_i^3}/{GM_e}}$, 
where \( G \) is the gravitational constant, and \( M_e \) is the mass of Earth. Then, the position of the $i$-th satellite at time slot \( t \) in Cartesian coordinates is given by 
\begin{equation} \label{eq:orbit}
\begin{aligned}
&x^{S}_{i,t} = H_i \left(\cos (\omega_{i}^{t}+\nu_i) \cos \Omega_i-\sin (\omega_{i}^{t}+\nu_i) \cos \iota_i \sin \Omega_i \right), \\
&y^{S}_{i,t} = H_i \left(\cos (\omega_{i}^{t}+\nu_i) \sin \Omega_i+\sin (\omega_{i}^{t}+\nu_i) \cos \iota_i \cos \Omega_i \right), \\
&z^{S}_{i,t} = H_i \left( \sin (\omega_{i}^{t}+\nu_i) \sin \iota_i \right),
\end{aligned}
\end{equation}
where \( \omega_t = \omega_{\text{init}} + (t \cdot \pi_i \mod \tau_i) \), with \( \pi_i \) the angular velocity.

\par Given these parameters, the instantaneous position of the satellite can be expressed in a Cartesian coordinate system, and thus allowing the HAP to predict its trajectory over time. Moreover, the angular velocity and orbital period of the satellite can further aid in estimating the coverage duration and transition dynamics of the satellite.

\par At each time slot \( t \), the HAP evaluates whether to maintain its current satellite connection or switch to a new one. Let $ \mathbf{L} = \{l_t | t \in \mathcal{T}, l_t \in \mathcal{S}\} $ represent the sequence of selected satellites over the timeline $\mathcal{T}$. This decision sequence variable could determine the satellite handover frequency (\textit{i.e.}, the number of satellite handovers). Let $N_t$ be the number of satellite handovers at time slot $t$, and $N_t$ evolves as follows: 
\begin{align}
    N_{t+1}=\begin{cases} 
N_t, & \text{if }l_t = l_{t+1} \\
N_t + 1, & \text{otherwise}
\end{cases}. \label{frequency of satellite handover}
\end{align}

\subsection{Communication Model}
\par In this section, we provide a comprehensive description of the satellite-to-HAP and HAP-to-users transmission models.
\subsubsection{FSO-based Satellite-to-HAP Transmission Model}
In the designed SAGIN system, the transmission between the active LEO satellite and HAP is carried out through an FSO link, which benefits from a direct line-of-sight path and low interference under atmospheric conditions. As such, the FSO channel gain at time slot $t$ can be expressed as  
\begin{align}
    h_{FSO}(t)=h_lh_a(t),
\end{align}
where $h_l$ is the link loss component, which is given by 
\begin{align}
    h_l = \frac{1}{2} (G_T + G_R-A_{FS}-A_{ATM}-L_{loss}-M_S ),
\end{align}
with $G_T$, $G_R$, $A_{FS}$, $A_{ATM}$, $L_{loss}$, and $M_S$ representing the transmit antenna gain, receive antenna gain, free-space loss, atmospheric attenuation, lenses loss, and system margin, respectively~\cite{wu2024deep}. 
\par Moreover, the the fading parameter $h_a(t)$ is represented using the Gamma-Gamma probability distribution, which is a widely adopted statistical model specifically designed to capture the impacts of fading caused by turbulence in the atmosphere on optical wave propagation. Following this, the signal-to-noise ratio (SNR) at the HAP is given by 
\begin{align}
    \gamma_{HAP}(t)=\frac{P_S\eta_{OE}^2h_{EGC}^2(t)}{N_AN_q}\triangleq \overline{\gamma}_{HAP}h_{EGC}^2(t), 
\end{align}
%\gamma_{HAP}=\frac{P_S\eta_{OE}^2h_{EGC}^2(t)}{N_AN_q}\triangleq \overline{\gamma}_{HAP}h_{EGC}^2(t), 
where $P_S$ is the satellite transmit power, $\eta_{OE}$ represents the optical-to-electrical translation parameters, $N_A$ is the number of receiver apertures to collect optical signals, $h_{EGC}(t)=\sum\nolimits^{N_A}_{q=1}h_{FSO}(t)$ denotes the scalar channel fading of the receive aperture ensemble, $N_q$ is the noise power, and $\overline{\gamma}_{HAP}=\frac{P_S\eta_{OE}^2}{{N_AN_q}}$ denotes the average SNR. Thus, the transmission rate between the satellite and the HAP can be expressed as 
\begin{align}\label{eq:rate_FSO}
    R_{FSO}(t)=B_{FSO} \log_2(1+\gamma_{HAP}(t)),
\end{align}
where $B_{FSO}$ denotes the bandwidth of the FSO-based link. 
\par In this work, the HAP operates based on a selective decode-and-forward (DF) protocol~\cite{Bian2025}. In this case, if the SNR exceeds a predefined threshold $\Lambda_{th}^S$, the HAP can correctly decode the received signal and re-encode it to forward to the target user. Otherwise, the HAP cannot decode the received signal and remains silent. Therefore, $z_{HAP}$ is given by 
\begin{align}
    z_{HAP}[t] = \mathbf{1}\{\gamma_{HAP} \geq \Lambda_{th}^S\}.
\end{align} 
%$\overline{\gamma}_{HAP}=P_S\eta_{OE}^2/({N_AN_q})$

\subsubsection{RF-based HAP-to-Users Transmission Model}
\par In the second transmission phase, the HAP communicates with the ground users by using RF signals. Specifically, the HAP transmits data to multiple ground users simultaneously through orthogonal frequency division multiple (OFDM) technology. Moreover, we employ the Nakagami-$m$ fading to model the HAP-to-users channel. As such, the channel gain between the HAP and $j$-th user can be expressed as  
\begin{align}
    h_j(t) = C_j(t) g_j(t),
\end{align}
where $g_j(t)$ follows a Nakagami distribution, and $C_j(t)$ represents the large-scale channel component, which is given by 
\begin{align}
    C_j(t) = G_{H} &+ R_j^H \nonumber\\&+ \frac{1}{2}\big(20\lg \lambda_{RF} - 10\eta\lg d_j(t) - 20\lg (4\pi)\big),
\end{align}
where $G_{H}$ and $R_j^H$ are the transmit and receive antenna gains, respectively, $\eta$ is the path loss exponent, $\lambda_{RF}$ represents the RF carrier wavelength, and $d_j(t)$ denotes the transmission distance from the HAP to $j$-th user at time slot $t$. %$g_j(t)$ follows a Nakagami distribution with fading parameter $m_{j}$
\par Following this, the SNR for the $j$-th user is given by 
\begin{align}
    \gamma_j(t) = \frac{P_j\left| h_j(t) \right|^2}{\sigma_j^2},
\end{align}
where $P_j$ is the transmit power from the HAP to $j$-th user and $\sigma_j^2$ is the noise power spectral density. Thus, the achievable rate from the HAP to $j$-th user is given by  
\begin{align}\label{eq:rate_j}
    R_j(t)=B_j \log_2\left(1+\gamma_j(t)\right),
\end{align}
where $B_j$ represents the carrier bandwidth. To fully capture the data flow and timing in the system, the next subsection will detail the queueing and scheduling mechanism of HAP for packet buffering and transmission.

\subsection{HAP Queue Model}
\par We describe the queuing and scheduling mechanisms at the HAP, where received packets are temporarily buffered before HAP-to-user transmission. The HAP maintains a finite-capacity and ordered buffer queue \(q_t\) of length \(L_q\) at time slot \(t\), which stores packets received from satellites. At each time slot, the HAP selects one satellite from the available set and receives a batch of packets denoted by \(\mathcal{P}_i^{H}[t]\), where \(\mathcal{P}_i^{H}[t]\subseteq q_t\). Then, the newly received packets are appended to the tail of \(q_t\). If \(|q_t|+|\mathcal{P}_i^{H}[t]|>L_q\), the oldest packets at the head of \(q_t\) are dropped to satisfy the capacity constraint. Moreover, the HAP schedules packets  $\mathcal{P}_j^{U}[t]$ from the queue $q_t$ to $j$-th user at time slot $t$ according to the adopted policy. 
\par Since the HAP operates in a quasi-stationary mode at stratospheric altitudes with solar-powered energy supply, its hovering energy consumption is minimal~\cite{Abderrahim2024}. Therefore, we do not consider the energy consumption of HAP in our system.

\subsection{AoI Model}
\par In this work, we adopt AoI as a key performance metric to quantify the freshness of status information packet. We denote $\theta_i[t]$, $\delta_i[t]$ and $\Delta_{i,j}[t]$ as the AoIs of the $i$-th satellite at the LEO satellite, HAP and ground users at time slot $t$, respectively. Note that, although sensing data at LEO satellites may experience minor buffering or scheduling delays before transmission, such delays are typically homogeneous across satellites. As such, we model packet generation as a stochastic process at the satellite side. Following this, the satellite-side AoI is given by 
%Note that, although sensing data at LEO satellites may experience minor buffering or scheduling delays before transmission, such delays are typically homogeneous across satellites. Therefore, we simplify the model by assuming random packet generation at the satellite side. Specifically, AoI increases linearly over time and resets to zero upon new packet generation at the LEO satellite side, which can be expressed as
\begin{align}
    \theta_i[t] &=
    \begin{cases} 
    0, & \text{if } \Gamma_i[t] = 1 \\
    \theta_i[t-1] + 1, & \text{otherwise}
    \end{cases} ,
\end{align}
where $\Gamma_{i}[t] = 1$ represents that a data packet is generated at $i$-th satellite at time slot $t$. Moreover, the delay from the $i$-th satellite to the HAP is given by 
\begin{align}
    T_i = \frac{\sum_{p \in \mathcal{P}_i^{H}[t]} D_p}{{R_{FSO}}},
\end{align} 
where $D_p$ represents the data size of packet $p$, and $R_{FSO}$ is the satellite-to-HAP transmission rate defined in Eq.~\eqref{eq:rate_FSO}. Thus, the AoI at the HAP side is given by 
\begin{align}
    \delta_i[t] &=
\begin{cases} 
\theta_i[t - T_i] + T_i, & \text{if } l[t]=i, z_{\text{HAP}}[t] = 1 \\
\delta_i[t - 1] + 1, & \text{otherwise}
\end{cases} .
\end{align}
\par Similarly, we denotes $T_j = \sum \nolimits_{p \in \mathcal{P}_j^{U}[t]} D_p / R_{j}$ as the delay from the HAP to the $j$-th ground user, where $R_{j}$ is given in Eq.~\eqref{eq:rate_j}. As such, the user-side AoI can be expressed as 
\begin{align}
    \Delta_{i,j}[t] &=
\begin{cases} 
\delta_i[t - T_j] + T_j, & \text{if } \gamma_j \geq \Lambda_{th}^H\\
\Delta_{i,j}[t - 1] + 1, & \text{otherwise}
\end{cases}. \label{eq:aoi}
\end{align}

\section{Problem Formulation and Analyses}\label{sec:formulation}
\par In this section, we formulate an optimization problem to improve the downlink transmission in SAGIN. We first present the decision variables and optimization objectives, then formulate an optimization problem, and finally give the corresponding analysis. 
%We first highlight the main concern of the system, then present the decision variables and optimization objectives, and finally formulate an optimization problem and give the corresponding analysis. 

\subsection{Problem Formulation}
\par In this work, we utilize the HAP as a relay to mitigate delays caused by direct satellite-to-ground communication challenges and minimize the satellite handover frequency to mitigate ping-pong handover issues. As such, the considered system involves two goals, \textit{i.e.}, reducing the time average AoI of satellites, and reducing the satellite handover frequency. 

\par At any time slot $t \in \mathcal{T}$ , the AoI is influenced by the satellite handover frequency. As such, the AoI and satellite handover frequency are tightly coupled. Thus, these optimization objectives have conflicting correlations. Accordingly, the coupling of variables and mutual influence of objectives require a multi-objective optimization formulation. The following decision variables need to be jointly determined: \textit{(i)} $\mathbf{P} = \left\{ P_j(t) \mid j \in \mathcal{U}, \, t \in \mathcal{T} \right\}$, a vector consisting of continuous variables denotes the transmit powers of ground users over time slots for HAP-to-users transmission. 
\textit{(ii)} $\mathbf{L} = \{l_t | t \in \mathcal{T}, l_t \in \mathcal{S}\}$, a vector consisting of discrete variables represents the selected satellite over time slots. 
%a vector consisting of discrete variables represents the index of the selected satellite during the timeline. 
\subsubsection{Optimization Objective 1}
\par The primary objective is to improve the time average AoI of LEO satellites. As such, the first optimization objective is given by
\begin{align}
    f_1(\mathbf{P},\mathbf{L})=\frac{1}{T}\sum_{t\in \mathcal{T}}\sum_{i=1}^{N_S} \Delta_{i}[t].
\end{align}

\subsubsection{Optimization Objective 2}
\par To minimize the AoI, the HAP needs to select an appropriate satellite as the receiver. However, frequent satellite handover will lead to ping-pong handover issues and incur additional link costs. Therefore, the second optimization objective can be expressed as %from the satellite list 
\begin{align}
    f_2(\mathbf{L})=N_{\mathcal{T}}.
\end{align}

\par Following this, the corresponding joint optimization problem is formulated as follows: 
\begin{subequations}\label{opti}
\begin{align}
(P1) : \mathop{\min}\limits_{\mathbf{P}, \mathbf{L}} 
 & \ F=(f_1,f_2), \\
\text{ s.t. }
    &l_t\in \mathcal{S}, \quad \forall t\in \mathcal{T},\label{opti:sub1}\\
    &P_{\min}\leq P_j(t) \leq P_{\max}, \quad \forall j\in \mathcal{U}, \, \forall t \in \mathcal{T}, \label{opti:sub3}\\
    &\sum_{\forall j\in \mathcal{U}}^{}P_j(t) \leq P_{HAP}\quad \forall j\in \mathcal{U}, \, \forall t \in \mathcal{T}, \label{opti:sub4}
\end{align}
\end{subequations}
where Eq.~\eqref{opti:sub1} ensures that the selected satellite is from the set of visible satellites at each time slot. Moreover, Eqs.~\eqref{opti:sub3} and \eqref{opti:sub4} define the HAP transmit power constraints. 
%Eq.~\eqref{opti:sub1} specifies the satellite connectivity constraint,

\subsection{Problem Analyses}
From the above formulation, we can observe that the problem has the following characteristics: 
\subsubsection{Non-convexity}
\par According to the problem $(P1)$ presented in Eq.~\eqref{opti}, the continuous transmit power variables $\mathbf{P}$ and the discrete satellite connectivity decisions $\mathbf{L}$ are coupled. As a result, the feasible region of $(P1)$ is inherently nonconvex, making the problem challenging to solve using conventional optimization methods. 
\subsubsection{NP-hard}
\par When fixing the discrete variable $\mathbf{L}$, $(P1)$ reduces to a continuous nonlinear programming problem. Moreover, fixing the continuous variable $\mathbf{P}$ results in a combinatorial optimization over discrete variables. As such, the resulting optimization can be reduced to a mixed-integer nonlinear programming (MINLP) problem, which has been proven to be an NP-hard problem~\cite{Belotti2013}. Therefore, solving $(P1)$ optimally is computationally intractable. 
%Due to the non-convexity and combinatorial complexity of the problem, $(P1)$ is a mixed-integer nonlinear programming (MINLP) problem, which is known to be NP-hard~\cite{Li2024}. Therefore, solving $(P1)$ optimally is computationally intractable. 
\subsubsection{Trade-off}
\par It can be observed that conflicts exist between the two optimization objectives. For instance, minimizing the AoI requires frequent satellite handovers to maintain freshness of information. In contrast, reducing handover frequency to save handover costs can lead to stale information and higher AoI. Therefore, a trade-off arises between maintaining low AoI and limiting handover frequency, which requires a balanced optimization strategy.

\section{The Proposed DD3QN-AS}\label{sec:DRL}
\par This section presents a DM-enhanced DRL-based method to address the formulated optimization problem. In pursuit of this objective, we first transform our problem into a Markov decision process (MDP). Then, we detail the proposed DD3QN-AS algorithm with several enhancements.

\subsection{MDP Formulation}
%\par The joint optimization problem in our scenario can be formulated as a MDP, which can be represented by a tuple \( (S, A, P, R, \gamma) \), where $S$, $A$, $P$, $R$, and $\gamma$ denote state space, action space, state transition probability, reward function, and discount factor, respectively. 
\par Following the aforementioned analysis, the formulated problem is inherently dynamic and stochastic, involving continuously changing system parameters such as the positions of LEO satellites, and user channel conditions. This results in a highly uncertain environment where decisions regarding satellite handover and power allocation must be made sequentially in real-time. Traditional optimization methods, such as convex optimization and evolutionary algorithms, face difficulties in such rapidly changing environments due to their reliance on accurate system models~\cite{Li2024}. In such cases, DRL learns adaptive policies directly from environment interactions, which efficiently handles high-dimensional states, nonlinear dynamics, and uncertain transitions via deep neural function approximation. Therefore, based on DRL frameworks, we can develop a robust algorithm that dynamically adjusts decision variables to maximize overall system performance. 
%Therefore, we adopt a DRL-based framework to jointly optimize satellite selection and power allocation for enhanced system performance. 
\par To solve this problem via DRL-based frameworks, we first reformulate the optimization problem shown in Eq.~\eqref{opti} as an MDP. Mathematically, an MDP is a tuple \( (S, A, P, R, \gamma) \), which are the state space, action space, state transition probability, reward function, and discount factor, respectively. Among them, the state, action, and reward are the most important components, which are detailed as follows: 
\begin{itemize}
    \item \textit{State Space:} The state space is designed to capture essential spatial and operational dynamics that influence system performance. Specifically, the positions of both satellites and ground users, along with the satellite-side AoI, HAP-side AoI, and user-side AoI, are incorporated to reflect spatial dynamics. Thus, the state at time slot $t$ is formally expressed as 
    \begin{align}
        s_t=\{ \boldsymbol{c}_t^S, \boldsymbol{c}_t^U, \theta[t], \delta[t],\Delta[t]\}.
    \end{align}
    %the local-averaged SNR between each candidate satellite and HAP,\gamma_{HAP}(t), 
    \item \textit{Action Space:} In our system, the HAP is required to select one satellite to connect with at any given time slot. In addition, the HAP allocates its available transmission power among ground users within that time slot. 
    Therefore, the possible action at time slot $t$ for the HAP is given by 
    \begin{align}
        a_t = \{l_t, \mathbf{P}\}.
    \end{align}
    \item \textit{Reward Function:} An appropriate reward function is critical for effective problem-solving in DRL frameworks. As such, the reward function involves our optimization objective and the corresponding constraints shown in Eq. \eqref{opti}. Specifically, our reward is defined as 
    \begin{align}
        r_t= -\rho_1 \sum_{i\in \mathcal{S}}\Delta_i[t] - \rho_2 N_t + \rho_3 \sum_{j\in \mathcal{U}}R_j,
    \end{align}
    where $\rho_1$, $\rho_2$, $\rho_3$ are three normalization parameters to adjust to bring these terms to the same order of magnitude. 
\end{itemize}

\subsection{Standard D3QN Algorithm}
\par In this paper, we adopt the dueling double deep Q-network (D3QN) algorithm as our optimization framework~\cite{Wang2016}. Specifically, D3QN is a model-free and off-policy reinforcement learning approach that efficiently mitigates value overestimation bias and improves the approximation of the value function, thereby leading to more stable and effective learning in complex environments~\cite{Wang2016}. D3QN integrates two important extensions of the original deep Q-Network (DQN), which are the double Q-learning and dueling network architecture~\cite{Cai2023}.

\par Similar to the standard DQN, D3QN aims to learn an optimal action value function \( Q(s, a) \), which represents the expected cumulative reward when taking action \( a \) in state \( s \) and following the optimal policy thereafter. In standard Q-learning, \( Q(s, a) \) is updated iteratively by using the Bellman equation, which can be expressed as 
\begin{align}
    Q(s, a) \leftarrow Q(s, a) + \alpha \left[ r + \gamma \max_{a'} Q(s', a') - Q(s, a) \right],
\end{align}
where \( \alpha \) is the learning rate, \( \gamma \) is the discount factor, \( r \) is the immediate reward, and \( s' \) is the next state. However, the \( \max \) operator is prone to overestimating action values. To mitigate this, double Q-learning decouples the selection and evaluation of the action by using two separate networks, which are \( Q_{\text{online}} \) and \( Q_{\text{target}} \). As such, the target value is given by 
\begin{align}\label{eq:target_value}
    y = r + \gamma Q_{\text{target}}(s', \arg\max_{a'} Q_{\text{online}}(s', a')).
\end{align}
%This modification mitigates overestimation bias by selecting actions through the online network and evaluating them through the target network. 
\par Furthermore, D3QN incorporates a dueling architecture that separates the estimation of the state-value function \( V(s) \) from the advantage function \( A(s, a) \). Then, the Q-value is given by 
\begin{align}\label{eq:Q-value}
    Q(s, a) = V(s) + \left( A(s, a) - \frac{1}{|\mathcal{A}|} \sum\nolimits_{a'} A(s, a') \right),
\end{align}
where \( \mathcal{A} \) denotes the action space. This decomposition provides more accurate $V(s)$ estimation under weak action influence. Moreover, the network parameters are optimized by minimizing the temporal-difference (TD) loss between predicted Q-values and bootstrapped targets, where the TD loss can be expressed as 
\begin{align}
    \mathcal{L}(\theta) = \mathbb{E}_{(s, a, r, s') \sim \mathcal{D}} \left[ \left( y - Q_{\text{online}}(s, a; \theta) \right)^2 \right],
\end{align}
where \( \theta \) denotes the parameters of the online network, and \( \mathcal{D} \) is the experience replay buffer. %As a result, D3QN improves both the stability and performance of the learning process. 
\par As such, D3QN enhances both the stability and performance of the learning process. Specifically, it mitigates overestimation bias through double Q-learning and improves learning efficiency by decoupling state value and action advantage via the dueling architecture. These enhancements lead to more accurate value estimation, faster convergence, and improved generalization, particularly in environments where the influence of individual action is non-uniform. 
%As such, D3QN improves both the stability and performance of the learning process. Moreover, D3QN reduces overestimation bias through double Q-learning and accelerates learning efficiency by decoupling state value and action advantage in the dueling architecture. This leads to more accurate value estimation, faster convergence, and better generalization in environments where the impact of actions varies. 

\begin{figure*}[t]
    \centering
    \includegraphics[width=0.9\linewidth]{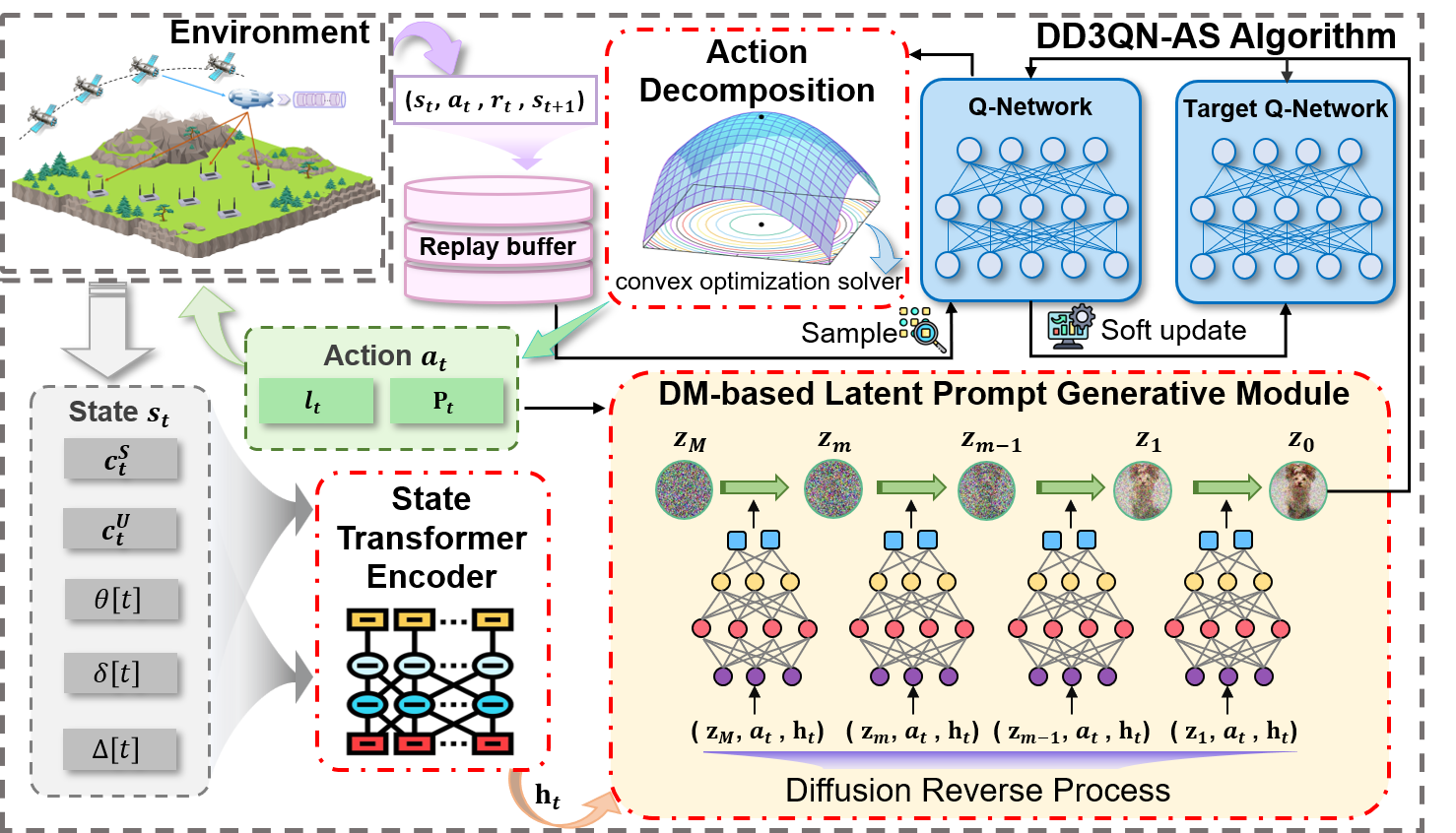}
    \caption{Framework of the proposed DD3QN-AS algorithm for the considered AoI-aware SAGIN, where the action decomposition module separates discrete and continuous actions, the STE captures temporal state features, and the DLPG generates robust action-state prompts to improve decision reliability. }
    \label{fig:algorithm}
\end{figure*}

\par However, in our MDP, the state space exhibits high variability and the reward signals are often sparse or delayed, which poses significant challenges for the standard D3QN algorithm. Specifically, the TD learning mechanism in D3QN struggles to accurately propagate value over extended time horizons, thereby leading to suboptimal credit assignment~\cite{ArjonaMedina2019}. Moreover, due to the limited expressiveness of the state-action encoding, the Q-value estimation tends to be unstable and sensitive to network initialization and exploration noise~\cite{Chi2024}. These issues ultimately hinder the convergence and reliability of the learned policy. Therefore, we propose several enhancements to D3QN adapted to the formulated optimization problem.

\subsection{DD3QN-AS Algorithm}
%Diffusion enhanced D3QN with State Transformer Encoder (DD3QN-AS)
%DLPG:improving the decision-making robustness of the agent under spatiotemporal fluctuations. 
\par To address the limitations of standard D3QN for the formulated problem, we propose DD3QN-AS, a novel algorithm that incorporates three key enhancements. Firstly, we incorporate an action decomposition mechanism that addresses the hybrid discrete-continuous action space. 
Secondly, we introduce a state transformer encoder (STE) mechanism that tokenizes the raw state into semantic representations, thus enabling more expressive and generalizable value estimation across temporal variations. Thirdly, we design a DM-based latent prompt generative (DLPG) module that leverages conditional denoising to refine latent state-action representations and to stabilize training under stochastic environment. 
\par Fig.~\ref{fig:algorithm} illustrates the architecture of the proposed DD3QN-AS algorithm, which highlights the interaction between the three enhancements and core D3QN learning components. The details of the three enhancements are as follows. %The details of the action decomposition, STE and DLPG are as follows. 
\subsubsection*{1) Action Decomposition}
\par To address the hybrid action space problem that combines continuous power allocation variables and discrete satellite decisions, we propose an action decomposition approach that separates these decision variables while maintaining solution optimality~\cite{Fan2019}. Specifically, as shown in Eq. \eqref{eq:aoi}, the transmission delay from HAP to ground users directly impacts AoI evolution. The delay is inversely related to the achievable data rate, which depends on the transmit power allocation decision variables as shown in Eq.~\eqref{eq:rate_j}. Therefore, instead of treating power allocation as part of the DRL action space, we can formulate it as a separate convex optimization problem as follows: 
\begin{subequations}\label{opti2}
\begin{align}
(P2) : \mathop{\max}\limits_{\mathbf{P}} \mathop{\min}\limits_{j\in \mathcal{U}} 
 & \{\ R_j=B_j\log_2(1+\frac{P_j\left| h_j \right|^2}{\sigma_j^2}) \}, \\
\text{ s.t. }
    &P_{\min}\leq P_j \leq P_{\max}, \label{eq:P2_1}\\
    &\sum\nolimits_{\forall j\in \mathcal{U}}^{}P_j \leq P_{HAP}, \label{eq:P2_2}
\end{align}
\end{subequations}

\begin{theorem}
\text{In the considered scenarios and feasible set of } $P_j$ (\(j \in \mathcal{U}\)), \text{ the problem $(P2)$ is convex.}
\end{theorem}

\begin{proof}
To prove that $(P2)$ is a convex optimization problem, it suffices to establish the following: 
\begin{enumerate}
    \item The objective function $\min_{j\in \mathcal{U}} R_j$ is concave with respect to $\mathbf{P}$;%The objective function
    \item The constraints in Eqs. \eqref{eq:P2_1} and \eqref{eq:P2_2} are convex.
\end{enumerate}
\par Firstly, the concavity of the objective function is established as follows. Define $a_j=|h_j|^2/\sigma_j^2$. To prove that $R_j(P_j)$ is concave, the second-order derivative is given by 
\begin{align}
    \frac{\partial^2 R_j}{\partial P_j^2}=-\frac{B_j}{\ln 2}\cdot \frac{a_j^2}{(1+a_jP_j)^2} <0.
\end{align}
Since $B_j > 0$, $a_j > 0$, and $(1+a_jP_j) > 0$ for all feasible $P_j$, the second derivative is strictly negative, which confirms that each $R_j(P_j)$ is strictly concave in $P_j$. 

\par Secondly, the convexity of the constraint set is established through the following observations. 
The constraints in Eq.~\eqref{eq:P2_1} define box constraints that form a hyperrectangle, which are convex. The constraint in Eq.~\eqref{eq:P2_2} is a linear inequality that defines a half-space, which is also convex. Since the intersection of convex sets is convex, the overall constraint set is convex. 
\par Therefore, problem $(P2)$ is a convex optimization problem.
\end{proof}

\subsubsection*{2) State Transformer Encoder}
\par To extract high-level spatiotemporal features from the environment state, we design an STE to replace the original MLP-based encoder in D3QN. Given a raw state vector $\mathbf{s}_t \in \mathbb{R}^d$, it is first segmented into a sequence of $n$ state tokens, which can be expressed as $\mathbf{X}_t = \left[
\mathbf{x}_t^{(1)}, \mathbf{x}_t^{(2)}, \dots, \mathbf{x}_t^{(n)}
\right], \quad \mathbf{x}_t^{(i)} \in \mathbb{R}^{d'}$. 

\par Moreover, these tokens are embedded and processed through a transformer encoder consisting of $L$ layers~\cite{Chen2021}. Each layer comprises two sub-layers, which are multi-head self-attention (MSA) mechanism and position-wise feed-forward network (FFN). The computation in the $l$-th layer is given by 

\begin{align}
    \tilde{\mathbf{H}}_t^{(l)} 
&= \mathbf{H}_t^{(l-1)} + \text{MSA}\big(\mathbf{H}_t^{(l-1)}\big) \notag \\
&= \mathbf{H}_t^{(l-1)} + 
\underbrace{
\text{softmax}\left( \mathbf{Q}^{(l)} \mathbf{K}^{(l) \top} /{\sqrt{d_k}} \right)
\mathbf{V}^{(l)},
}_{\text{Multi-head self-attention}}
\\[6pt]
\mathbf{H}_t^{(l)} 
&= \tilde{\mathbf{H}}_t^{(l)} + \text{FFN}\big(\tilde{\mathbf{H}}_t^{(l)}\big) \notag \\
&= \tilde{\mathbf{H}}_t^{(l)} + 
\underbrace{
\text{GELU}\big(\tilde{\mathbf{H}}_t^{(l)} \mathbf{W}_1^{(l)} + \mathbf{b}_1^{(l)}\big) \mathbf{W}_2^{(l)} + \mathbf{b}_2^{(l)},
}_{\text{Feed-forward network}}
\end{align}
where $d_k$ is the dimension of the key vectors, and $\mathbf{Q}^{(l)} = \mathbf{H}_t^{(l-1)} \mathbf{W}_Q^{(l)}$, $\mathbf{K}^{(l)} = \mathbf{H}_t^{(l-1)} \mathbf{W}_K^{(l)}$, and $\mathbf{V}^{(l)} = \mathbf{H}_t^{(l-1)} \mathbf{W}_V^{(l)}$ are the projections of the query, key, and value for the $l$-th layer, respectively. 
After the final layer, a pooling operation is applied to obtain the compact state representation, which can be expressed as 
\begin{align}
    \mathbf{h}_t = \text{Pool}(\mathbf{H}_t^{(L)}) = \frac{1}{n}\sum_{i=1}^n \mathbf{H}_t^{(L, i)}.
\end{align} 

%引用：Reasoning with latent diffusion in offline reinforcement learning
\subsubsection*{3) DM-based Latent Prompt Generative Module}
\par To improve policy robustness in the presence of non-stationarity, we introduce a DLPG module to stochastically augment the encoded state $\mathbf{h}_t$ extracted by STE. Specifically, DLPG adopts a conditional generative modeling strategy based on diffusion processes~\cite{Rombach2022}\cite{Shribak2024}. Let $\mathbf{z}_0$ be an unobserved latent prompt drawn from a standard Gaussian prior. Then, the conditional reverse process can be expressed as 
\begin{align}
    p_\theta(\mathbf{z}_{0:M} \mid \mathbf{h}_t, \mathbf{e}_a) = p(\mathbf{z}_M) \prod_{m=1}^{M} p_\theta(\mathbf{z}_{m-1} \mid \mathbf{z}_m, \mathbf{h}_t, \mathbf{e}_a),
\end{align}
where $\mathbf{e}_a$ is the embedding of action $a_t$, $m$ denotes each reverse step, and $M$ is the total number of denoising steps. 
%$\mathbf{z}_M \sim \mathcal{N}(0, \mathbf{I})$,
\par To train the denoising model, the objective is defined as minimizing the noise reconstruction error, which is given by 
%The training objective minimizes the prediction error of added noise across time steps, leading to the simplified loss:
\begin{align}\label{eq:loss_dlpg}
    \mathcal{L}_{\text{DLPG}} = \mathbb{E}_{m, \mathbf{z}_0, \epsilon} \left[ \left\| \epsilon - \epsilon_\theta(\mathbf{z}_m, m \mid \mathbf{h}_t, \mathbf{e}_a) \right\|^2 \right].
\end{align}
During inference, $\hat{\mathbf{z}}_0$ is obtained by iterative sampling from the learned denoising process. Moreover, the enhanced latent representation is represented as $\tilde{\mathbf{h}}_t = \text{Concat}(\mathbf{h}_t, \hat{\mathbf{z}}_0)$, which serves as input to the dueling Q-network $Q(s_t, a_t) = f_Q(\tilde{\mathbf{h}}_t, a_t)$. In addition, the detailed generation procedure is summarized in Algorithm~\ref{alg:DLPG}. 
%这里可以不用伪代码，画两个图更稳妥一点
\begin{algorithm}[t]
\label{alg:DLPG}
\caption{DLPG Module}
\KwIn{Encoded state $\mathbf{h} \in \mathbb{R}^{d_h}$, action $a$, diffusion steps $M$, noise schedule $\{\beta_m\}_{m=1}^M$}
\KwOut{Enhanced latent prompt $\hat{\mathbf{z}}_0$}
\tcc{Training phase } %(for $\epsilon_\theta$)
Sample noise $\epsilon \sim \mathcal{N}(0, \mathbf{I})$ \;
Sample diffusion timestep $m\in\{1, \ldots, M\}$ \;
Sample latent $\mathbf{z}_0 \sim \mathcal{N}(0, \mathbf{I})$ \;
Compute noisy latent: $\mathbf{z}_m = \sqrt{\alpha_m} \mathbf{z}_0 + \sqrt{1-\alpha_m} \epsilon,\quad \alpha_m = \prod_{i=1}^m (1 - \beta_i)$ \;
Embed action $a_t$ to vector $\mathbf{e}_a$ \;
Condition on both $\mathbf{h}$ and $\mathbf{e}_a$: $\hat{\epsilon}_\theta = \epsilon_\theta(\mathbf{z}_m, m \mid \mathbf{h}, \mathbf{e}_a)$ \;
Compute loss using Eq.~\eqref{eq:loss_dlpg}\; %$\mathcal{L}_{\text{DLPG}} = \|\epsilon - \hat{\epsilon}_\theta\|^2$ \;
\tcc{Inference phase}
Initialize $\mathbf{z}_M \sim \mathcal{N}(0, \mathbf{I})$ \;
\For{$m = M$ \KwTo $1$}{
    Predict noise: $\hat{\epsilon}_\theta = \epsilon_\theta(\mathbf{z}_m, m \mid \mathbf{h}, \mathbf{e}_a)$ \;
    Compute posterior mean: $\mathbf{z}_{m-1} = \frac{1}{\sqrt{1-\beta_m}}\left(\mathbf{z}_m - \frac{\beta_m}{\sqrt{1-\alpha_m}} \hat{\epsilon}_\theta\right) + \sigma_m \mathbf{w}, \mathbf{w} \sim \mathcal{N}(0, \mathbf{I})$ \;
}
\Return final enhanced latent prompt $\hat{\mathbf{z}}_0 = \mathbf{z}_0$
\end{algorithm}

\begin{algorithm}[t]
\label{alg:whole-alg}
\caption{DD3QN-AS Algorithm}
\KwIn{Environment $Env$, total episodes $N_{\text{eps}}$, batch size $B$, soft update rate $\tau$, diffusion steps $M$}
\KwOut{Q-network parameters $\theta$}

Initialize transformer encoder $\mathcal{F}_{\text{Enc}}$, replay buffer $\mathcal{D}$, Q-network $\theta$, target network $\theta^- \leftarrow \theta$\;
%Initialize Transformer encoder $\mathcal{F}_{\text{Enc}}$\;

\For{$\text{episode} = 1$ \KwTo $N_{\text{eps}}$}{
    Reset environment and observe initial state $s_0$\;
    \For{$t = 1$ \KwTo $T$}{
        Encode state using STE: $h_t \leftarrow \mathcal{F}_{\text{Enc}}(\text{Tokenize}(s_t))$\;
        \uIf{random() $< \epsilon$}{
            Select random action $a_t$\;
        }
        \Else{
            $a_t = \arg\max_a Q(h_t, a)$\;
        }
        Solve power allocation via convex optimization: $p_t = \text{ConvexSolver}(a_t, s_t)$\;
        Execute $(a_t, p_t)$, observe reward $r_t$, next state $s_{t+1}$, done flag\;
        Store $(s_t, a_t, r_t, s_{t+1}, \text{done})$ into buffer $\mathcal{D}$\;
        \If{$|\mathcal{D}| > B$}{
            Sample batch $(s, a, r, s', \text{done}) \sim \mathcal{D}$\;
            Encode states: $\mathbf{h}_t = \mathcal{F}_{\text{Enc}}(\text{Tokenize}(s))$,\
            $\mathbf{h}_t' = \mathcal{F}_{\text{Enc}}(\text{Tokenize}(s'))$\;
            Compute latent representations via DLPG: $\mathbf{z}_{0} \leftarrow \text{DLPG}(\mathbf{h}_t, a, M)$\tcp*{Algorithm~\ref{alg:DLPG}}
            %$z'_{0} \leftarrow \text{DLPG}(\mathbf{h}_t', \arg\max\limits_{a'} Q(\mathbf{h}_t', a'), M)$ \tcp*{Algorithm~\ref{alg:DLPG}}
            Compute target Q-value using Eq.~\eqref{eq:target_value}\; %$y = r + \gamma (1 - \text{done}) \cdot Q_{\theta^-}(z'_{\text{enhanced}}, a^*)$ \;
            Compute current Q-value using dueling structure according to Eq.~\eqref{eq:Q-value}\; %$Q(z_{\text{enhanced}}, a) = V(z_{\text{enhanced}}) + \left(A(z_{\text{enhanced}}, a) - \frac{1}{|\mathcal{A}|} \sum_{a'} A(z_{\text{enhanced}}, a') \right)$\;
            Minimize loss: $\mathcal{L} = \| Q(\mathbf{z}_{0}, a) - y \|^2$\;
            Update $\theta$ via gradient descent\;
            Soft update target: $\theta^- \leftarrow \tau \theta + (1 - \tau) \theta^-$\;
        }

        \If{$\text{done}$}{
            \textbf{break}\;
        }
    }
}
\Return{$\theta$}
\end{algorithm}

\subsection{The Main Steps and Complexity Analyses of the proposed DD3QN-AS}
\par As such, the main steps of the DD3QN-AS algorithm are shown in Algorithm~\ref{alg:whole-alg}, and the complexity analyses are detailed as follows. Specifically, we analyze the computational and space complexity of the proposed DD3QN-AS algorithm during the training and execution phases. 
%\par In this section, we analyze the computational and space complexity of DD3QN-AS during training and execution phases.
\subsubsection{Training Phase}
The computational complexity of DD3QN-AS in the training phase is given by $O\left(N_{\text{eps}} T\cdot \left(|\theta_{\text{STE}}| + |\theta_Q| + M \cdot |\theta_{\text{DLPG}}| + C_{\text{opt}}\right)\right)$, which is summarized as follows: 
\begin{itemize}
    \item \textit{Network Initialization}: This phase involves the initialization of network parameters. The computational complexity is $O\left(|\theta_{\text{STE}}| + |\theta_Q| + |\theta_{\text{DLPG}}|\right)$, where $|\theta_Q|$ denotes the number of parameters in the Q-network, $|\theta_{\text{STE}}|$ represents the parameters in the transformer encoder, and $|\theta_{\text{DLPG}}|$ represents the parameters in the DLPG module. 
    \item \textit{Action Selection}: In this phase, states are encoded via the transformer encoder and actions are selected based on Q-values at each time step~\cite{Wang2016}~\cite{Chen2021}. The complexity is $O\left(N_{\text{eps}} T (|\theta_{\text{STE}}| + |\theta_Q|)\right)$, where $N_{\text{eps}}$ is the number of episodes and $T$ is the maximum time steps per episode. 
    \item \textit{Network Update}: During training updates, the diffusion process is the most computationally intensive due to iterative denoising over $M$ steps~\cite{Shribak2024}. The complexity of this phase is calculated as $O\left(\frac{N_{\text{eps}} T B (|\theta_{\text{STE}}| + M \cdot |\theta_{\text{DLPG}}| + |\theta_Q|)}{B}\right) = O\left(N_{\text{eps}} T (|\theta_{\text{STE}}| + M \cdot |\theta_{\text{DLPG}}| + |\theta_Q|)\right)$. %, where $B$ is the batch size. 
    \item \textit{Convex Optimization}: Each step involves solving a convex optimization problem for power allocation, which contributes an additional cost $C_{\text{opt}}$ per step~\cite{Liu2025a}. %
\end{itemize}

\par In the training phase, the space complexity of DD3QN-AS algorithm is determined by the size of the neural network parameters and replay buffer, which is $O\left(|\theta_{\text{STE}}| + |\theta_Q| + |\theta_{\text{DLPG}}| + D(2|s| + |a| + 2)\right)$, where $D$ represents the capacity of the replay buffer $|s|$ and $|a|$ denote the dimensions of the state and action spaces, respectively. 

\subsubsection{Execution Phase}
\par In the execution phase, the computational complexity of the DD3QN-AS algorithm is $O\left(|\theta_{\text{STE}}| + |\theta_Q|\right)$, as only the transformer encoder and Q-network are required for state encoding and action selection. Similarly, the space complexity during execution is $O\left(|\theta_{\text{STE}}| + |\theta_Q|\right)$. Since the diffusion component, target network, and replay buffer are not involved during the execution phase, the space and computational overhead are minimized, which makes the process more efficient during inference. 

\section{Simulation Results and Analysis}\label{sec:simulation}
%\par In this section, we detail the simulation results and analyses. We first introduce the simulation setting and benchmarks, and then provide the corresponding results. 
\par In this section, we present comprehensive evaluations of our proposed approach and verify the effectiveness and robustness of the proposed DD3QN-AS in addressing the formulate optimization problem under various settings. 
%\par In this section, we detail the simulation results and analyses. First, we introduce the simulation setting and benchmarks. Then we present the corresponding results.

\subsection{Simulation Setups}
\subsubsection{Environmental Details}
\par We consider a multi-tier hybrid AoI-aware downlink communication system for SAGINs, which consists of a LEO satellite constellation, an HAP, and multiple ground users. Specifically, the constellation comprises $10$ LEO satellites deployed at altitudes ranging from $5 \times 10^5$ m to $1.8 \times 10^6$ m. These satellites operate in diverse orbital planes, including equatorial, mid-inclination, and near-polar orbits. %Moreover, minor phase and right ascension offsets are applied to diversify orbital positioning. 
Moreover, the HAP hovers at a constant altitude of $20000$ m, maintaining a geostationary position relative to the ground. In addition, we consider 10 ground users randomly distributed within a square area of 1 \text{km} $\times$ 1 \text{km}. Following this, the satellites-to-HAP communication link utilizes FSO link, and the RF links connect the HAP to the ground users, and these communication-related parameters follow~\cite{wu2024deep}. 
\subsubsection{Model Design}

\par The proposed DD3QN-AS algorithm extends the standard D3QN by integrating an STE and a DLPG module. Specifically, the STE leverages a two-layer transformer encoder with four heads and 64-dimensional embeddings to capture temporal structure from tokenized states. The DLPG adopts sinusoidal time embeddings and prompt tokens of size 32 to denoise over multiple diffusion steps, and the dueling Q-network uses a three-layer MLP (256, 256, 128 neurons) with ReLU activation. Moreover, the model is trained using the Adam optimizer~\cite{Kingma2015} with a learning rate of $3 \times 10^{-4}$. Additional training hyperparameters are outlined in Table~\ref{tab:hyperparams}. 
%Scenario and Algorithm Setups
\begin{table}[]
\centering
\caption{Training Hyperparameters for DD3QN-AS}
\label{tab:hyperparams}
\begin{tabular}{lll}
\toprule
\textbf{Parameter} & \textbf{Description} & \textbf{Value}\\
\midrule
$B$ & Batch size & $128$\\
$\gamma$ & Discount factor & $0.99$\\
$D$ & Capacity of the experience replay buffer & $1 \times 10^5$\\
$\tau$ & Soft update rate & $0.005$\\
$d$ & Frequency of policy updates &$10$\\
$M$ & Denoising steps for the diffusion model &$4$\\
$N_{\text{eps}}$ & Number of training episodes &$5000$\\
\bottomrule
\end{tabular}
\end{table}

\subsubsection{Baseline Algorithms}
\par To evaluate the effectiveness of the proposed DD3QN-AS algorithm, we compare it with the following policy-based baseline methods: 
\begin{itemize}
    \item \textit{Random Method}: A stochastic selection strategy that assigns equal probability to all candidate satellites and selects one at random at each decision step. 
    \item \textit{Enhanced Weighted Greedy (EWG) Method}: A heuristic strategy that selects the satellite with the lowest weighted sum of AoI, HAP buffer load, and handover frequency. %~\cite{Kaul2012}
    \item \textit{Round Robin (RR) Method}: A deterministic selection strategy that selects satellites in a fixed cyclic order and increases the selection index by one at each decision step. %~\cite{omar2021comparative}
\end{itemize}

\par Moreover, we introduce the following DRL methods as baselines, including the following algorithms:  
\begin{itemize}
    \item \textit{Proximal Policy Optimization (PPO)}~\cite{Cai2023}: An on-policy policy-gradient method that uses a clipped surrogate objective to stabilize training and constrain policy updates. 
    \item \textit{Soft Actor-Critic (SAC)}~\cite{Haarnoja2018}: An off-policy actor–critic algorithm that maximizes return and entropy, which improves exploration and stability. %
    \item \textit{Truncated Quantile Critics (TQC)}~\cite{Kuznetsov2020}: An off-policy method using quantile regression with truncated tails to mitigate value overestimation. 
    \item \textit{DQN}~\cite{Mnih2015}: A value-based algorithm that integrates Q-learning with deep neural networks for stable training.  %
    \item \textit{Double Deep Q-Network (DDQN)}~\cite{Hasselt2016}: An extension of DQN that mitigates Q-value overestimation using separate networks for action selection and evaluation. %
    \item \textit{Standard D3QN}~\cite{Wang2016}: A dueling-architecture variant of D2QN that decouples state value and action advantage to enhance convergence. 
\end{itemize}
\begin{figure}[t]
    \centerline{\includegraphics[width=\linewidth]{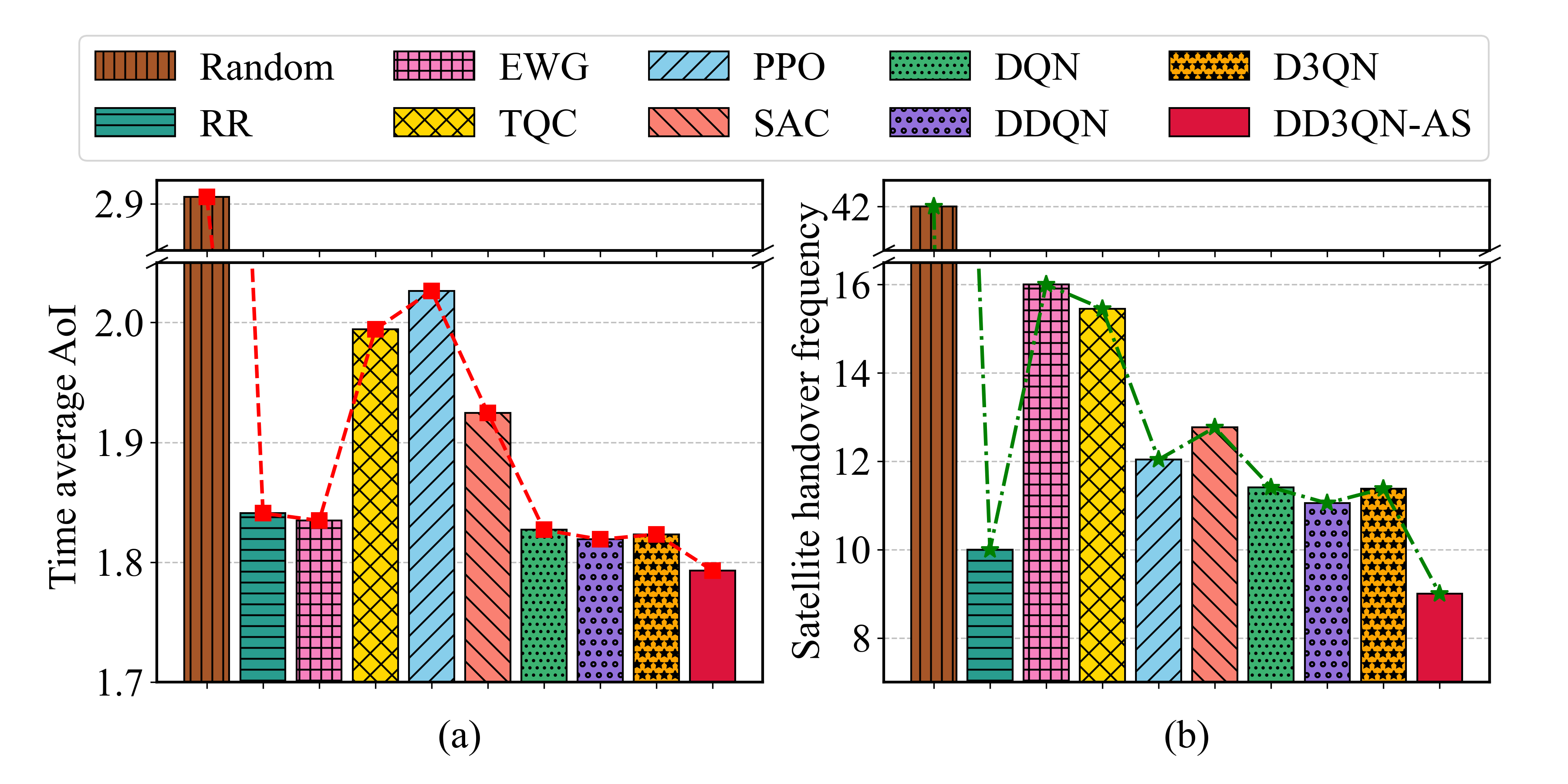}}
    \caption{Optimization objective values obtained by different methods.}
    \label{fig:total_reward}
\end{figure}
\begin{figure}[t]
    \centerline{\includegraphics[width=\linewidth]{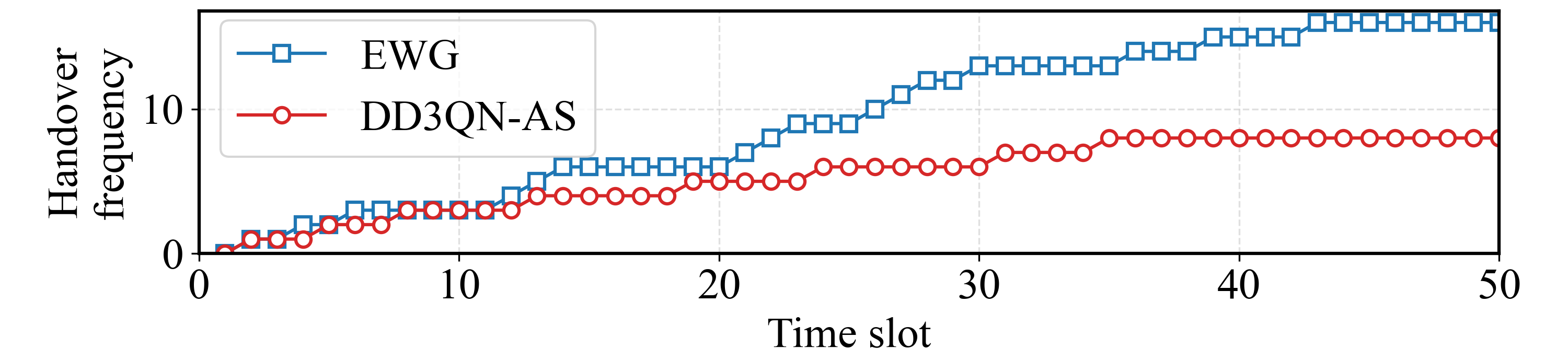}}
    \caption{Accumulated handover frequency over time slots.}
    \label{fig:handover_timeslot}
\end{figure}
\begin{figure*}[t]
\centering           
\subfloat[]
{
    \label{fig:reward}\includegraphics[width=0.33\textwidth]{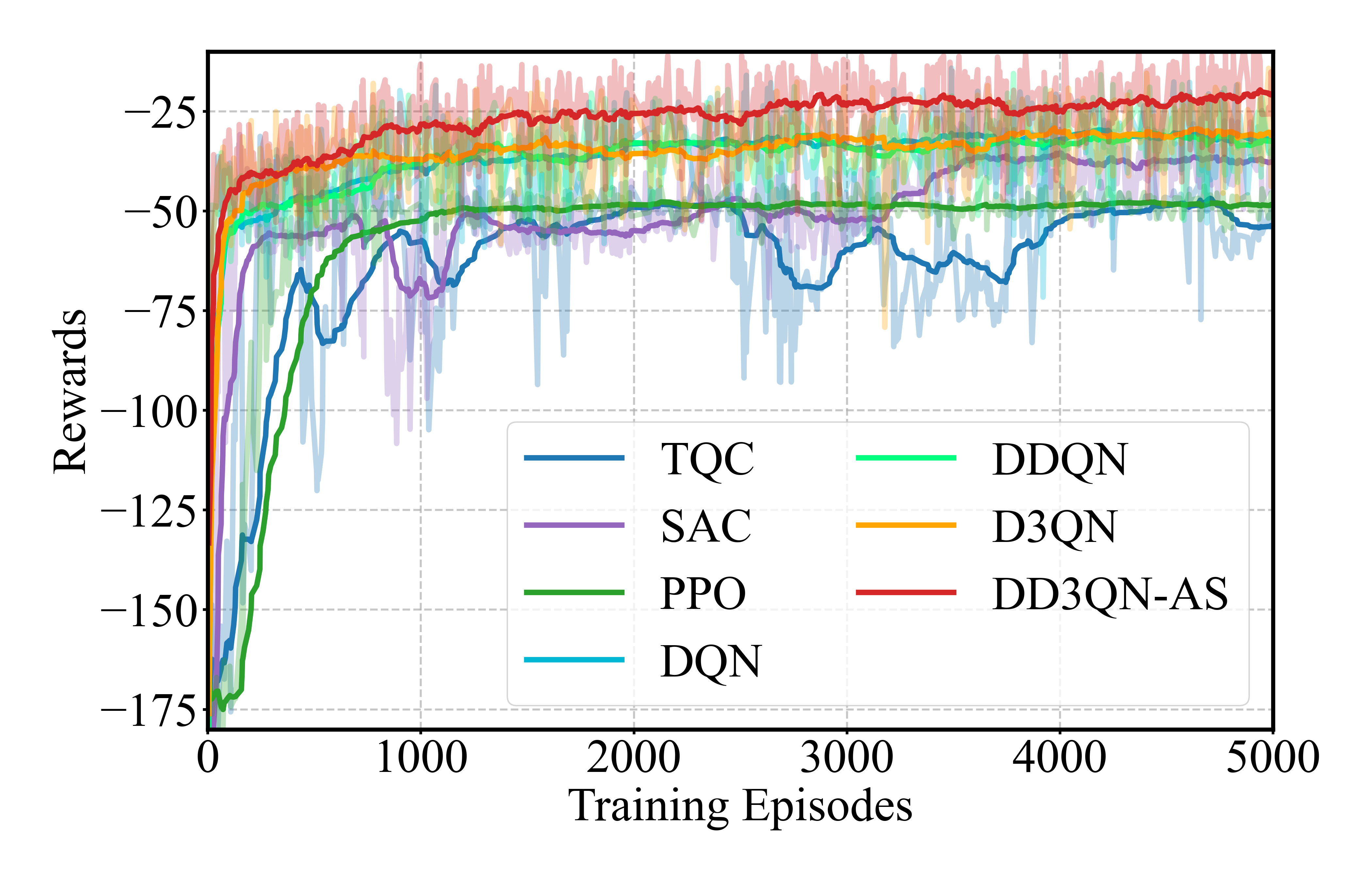}
}
\subfloat[]
{
    \label{fig:aoi}\includegraphics[width=0.33\textwidth]{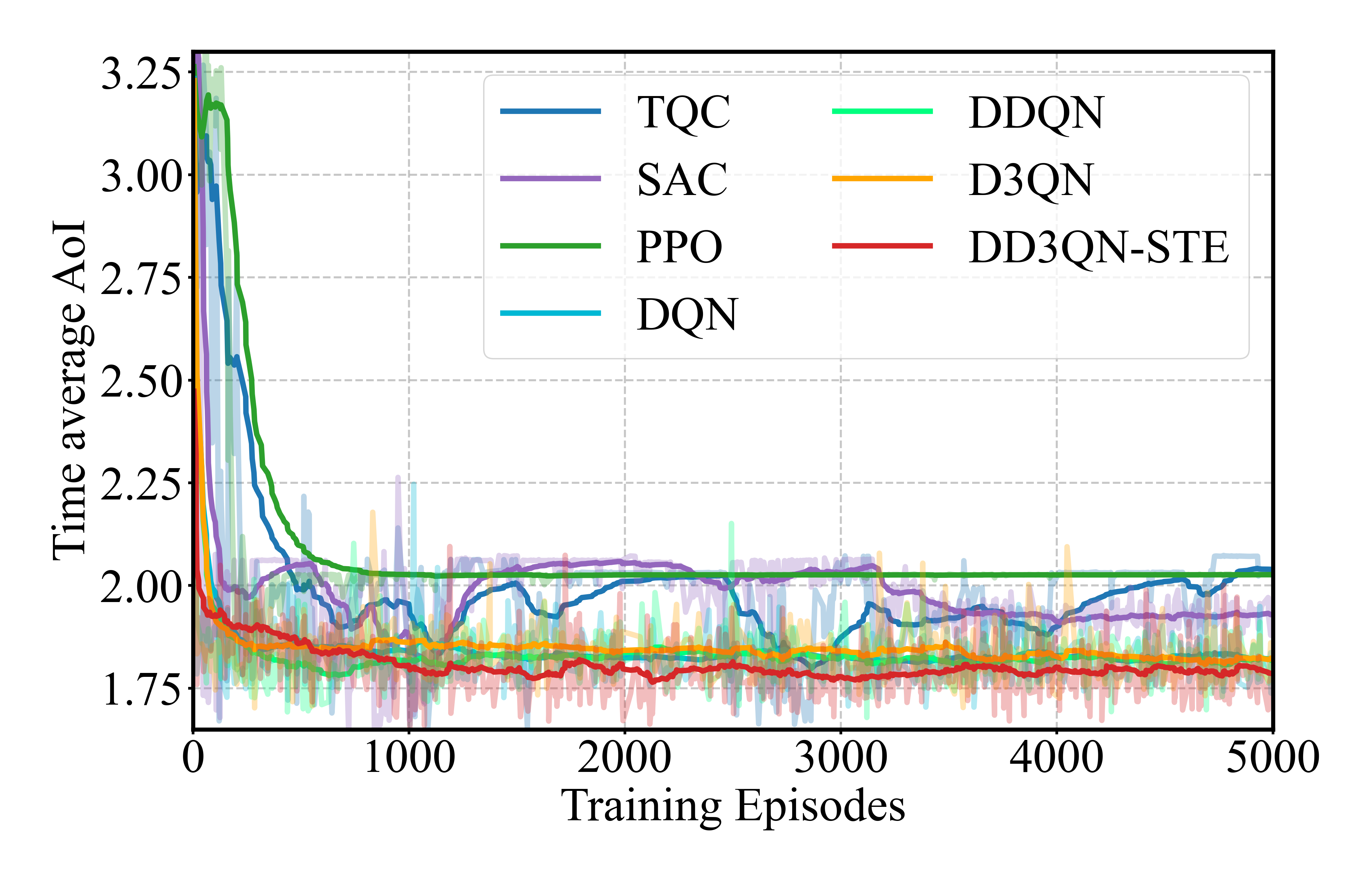}
}
\subfloat[]
{
    \label{fig:handover.}\includegraphics[width=0.33\textwidth]{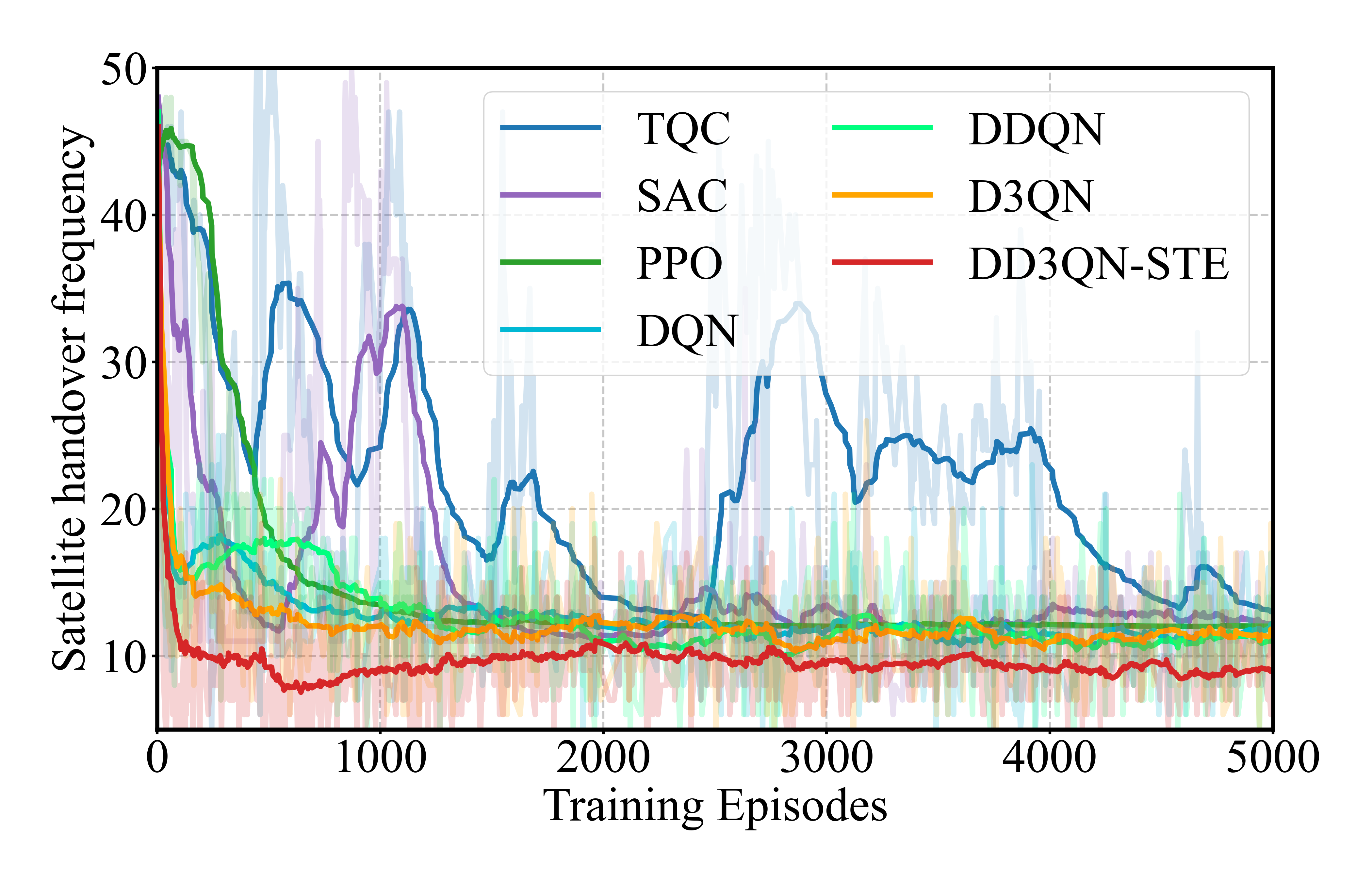}
}
\caption{Convergence comparison of different algorithms. (a) Rewards of different baselines. (b) Time average AoI of different baselines. (c) Handover frequency of different baselines.}  
\label{fig:Convergence comparison.}
\end{figure*}
\begin{figure}[t]
    \centerline{\includegraphics[width=\linewidth]{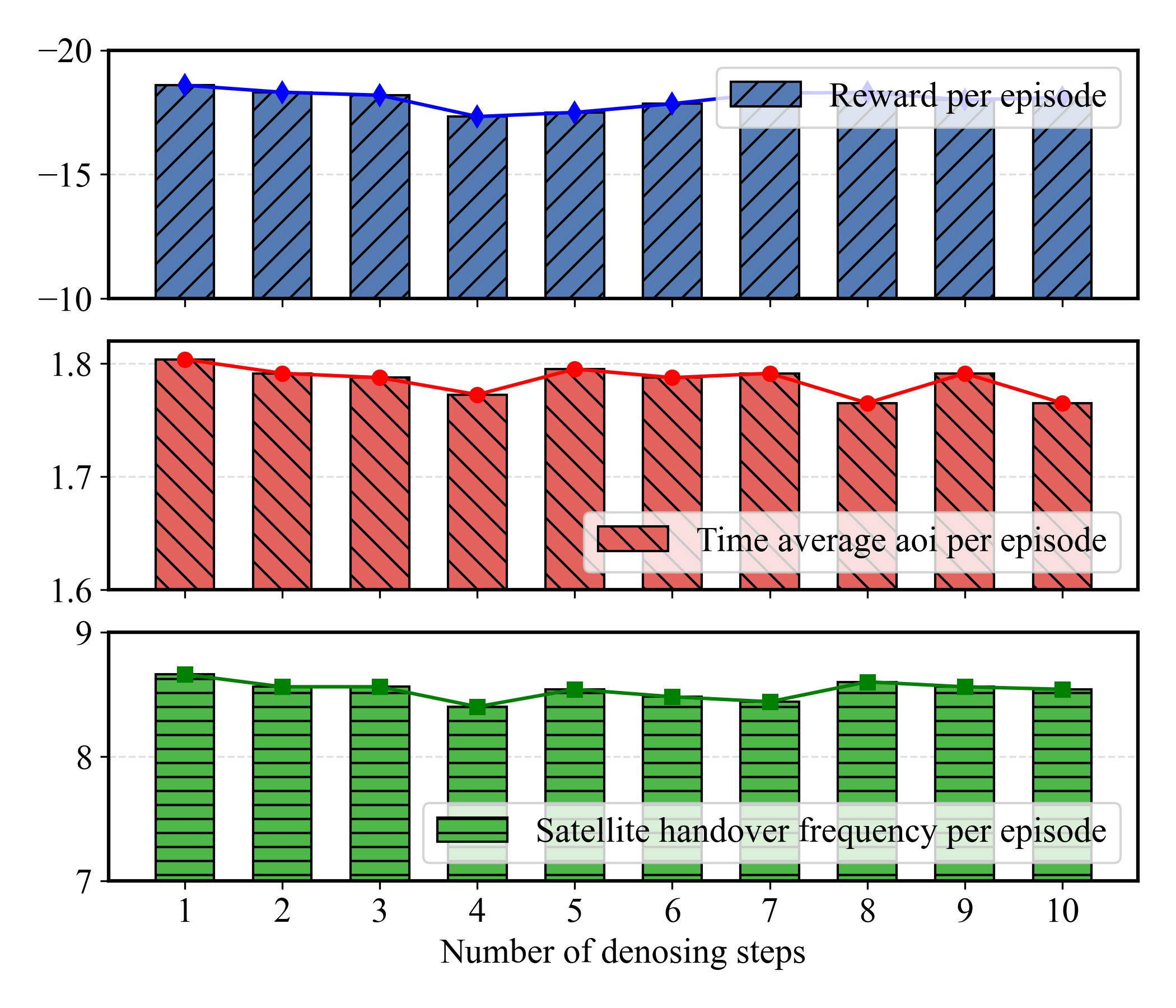}}
    \caption{Comparison of curves of DD3QN-AS with different denoising steps.}
    \label{fig:denosing_steps}
\end{figure}

\subsection{Optimization Results} 
%\subsection{Simulation Results}
\subsubsection{Algorithm Performance Evaluation}
\par Fig.~\ref{fig:total_reward} presents the optimization objective values achieved by different algorithms. As can be seen, the proposed DD3QN-AS algorithm outperforms the policy-based methods and conventional DRL algorithms in both time average AoI and handover frequency due to its improved environmental understanding and exploration capabilities. Specifically, compared to the best-performing baseline (\textit{i.e.}, DDQN), DD3QN-AS reduces the average AoI by approximately 1.7$\%$ and the handover frequency by 15$\%$. As such, the superior performance for minimizing AoI and handover frequency establishes DD3QN-AS as a practical and effective algorithm for the considered SAGIN system. 
%\subsubsection{Handover Stability Analysis}
\par Fig.~\ref{fig:handover_timeslot} shows the accumulated number of handovers versus time slots for the proposed DD3QN-AS algorithm and EWG method. We can observe that the DD3QN-AS algorithm performs a more gradual increase in handover frequency compared to EWG. This significant difference in handover accumulation patterns validates the effective handling of ping-pong effects via our learned handover policy. 

\subsubsection{Convergence Performance Comparison with Other DRL Benchmarks}
\par Fig.~\ref{fig:Convergence comparison.}(a) compares the convergence behaviors of all algorithms in terms of episodic reward. As can be seen, the proposed DD3QN-AS converges the fastest and achieves the highest steady-state reward with markedly reduced variance after the initial learning phase. Moreover, the standard D3QN converges more rapidly than other algorithms while saturates at a clearly lower level. 
This may be because that the STE encoder supplies richer spatiotemporal features to the dueling-double Q backbone, thereby improving value estimation under nonstationary conditions and stabilizing learning. %In addition, TQC shows volatility and the weakest convergence performance, possibly because its quantile update sensitivity to reward distribution shifts. 
\par Fig.~\ref{fig:Convergence comparison.}(b) illustrates the convergence behaviors for optimization objective $f_1$, \textit{i.e.}, minimizing the time average AoI of LEO satellites. As can be seen, all the algorithms exhibit convergence for $f_1$. Moreover, the proposed DD3QN-AS achieves and maintains the lowest AoI. D3QN follows but remains higher. Furthermore, other baselines level off at noticeably larger AoI values. This superior performance can be attributed to the STE mechanism that extracts rich spatiotemporal features through transformer-based token processing and the DLPG module that enhances representation robustness under stochastic channel conditions, jointly enabling more precise transmission timing decisions crucial for AoI minimization. 
\par Fig.~\ref{fig:Convergence comparison.}(c) indicates the convergence behaviors for optimization objective $f_2$, \textit{i.e.}, minimizing satellite handover frequency. We can observe that DD3QN-AS converges to a relatively low and stable handover frequency. In contrast, TQC exhibits significant volatility with recurring spikes in handover frequency. This may because that the ability of STE to capture temporal dependencies in satellite movement patterns and channel variations, while the DLPG module contributes to maintaining consistent decision-making under stochastic conditions. In this case, the proposed DD3QN-AS significantly reduces unnecessary handovers compared to benchmark algorithms. 

\subsection{Effect of Different Denoising Steps}
%Impact of Algorithm Parameters
\par The number of denoising steps in the diffusion reverse process constitutes a key parameter that can significantly impact the performance of DM. Firstly, the denoising steps determine the efficiency of DM in reducing noise and producing high-quality samples. Secondly, an increase in denoising steps also leads to longer training time~\cite{Zhang2025}. Therefore, we compare the impact of varying the number of denoising steps on the performance of DD3QN-AS. 
As shown in Fig.~\ref{fig:denosing_steps}, increasing the number of denoising steps in the DLPR module generally leads to improved performance. This improvement can be attributed to the enhanced denoising capacity, which enables the model to better reconstruct the original state representation from noisy perturbations, thereby guiding more robust Q-value estimation through the auxiliary loss. However, we observe that beyond a certain threshold, specifically four steps in our algorithm, the marginal benefits begin to diminish. This is likely due to the overfitting of the model to the noise distribution, where excessive denoising steps result in oversmoothing of latent representations. Consequently, this can suppress fine-grained distinctions among critical state-action pairs, thereby negatively impacting action diversity and long-term decision quality. As such, a moderate setting strikes a balance between noise robustness and representation fidelity, while also mitigating the computational overhead introduced by repeated denoising. %These results underscore the importance of carefully selecting the number of denoising steps in DM. 
\begin{figure}[t]
    \centerline{\includegraphics[width=\linewidth]{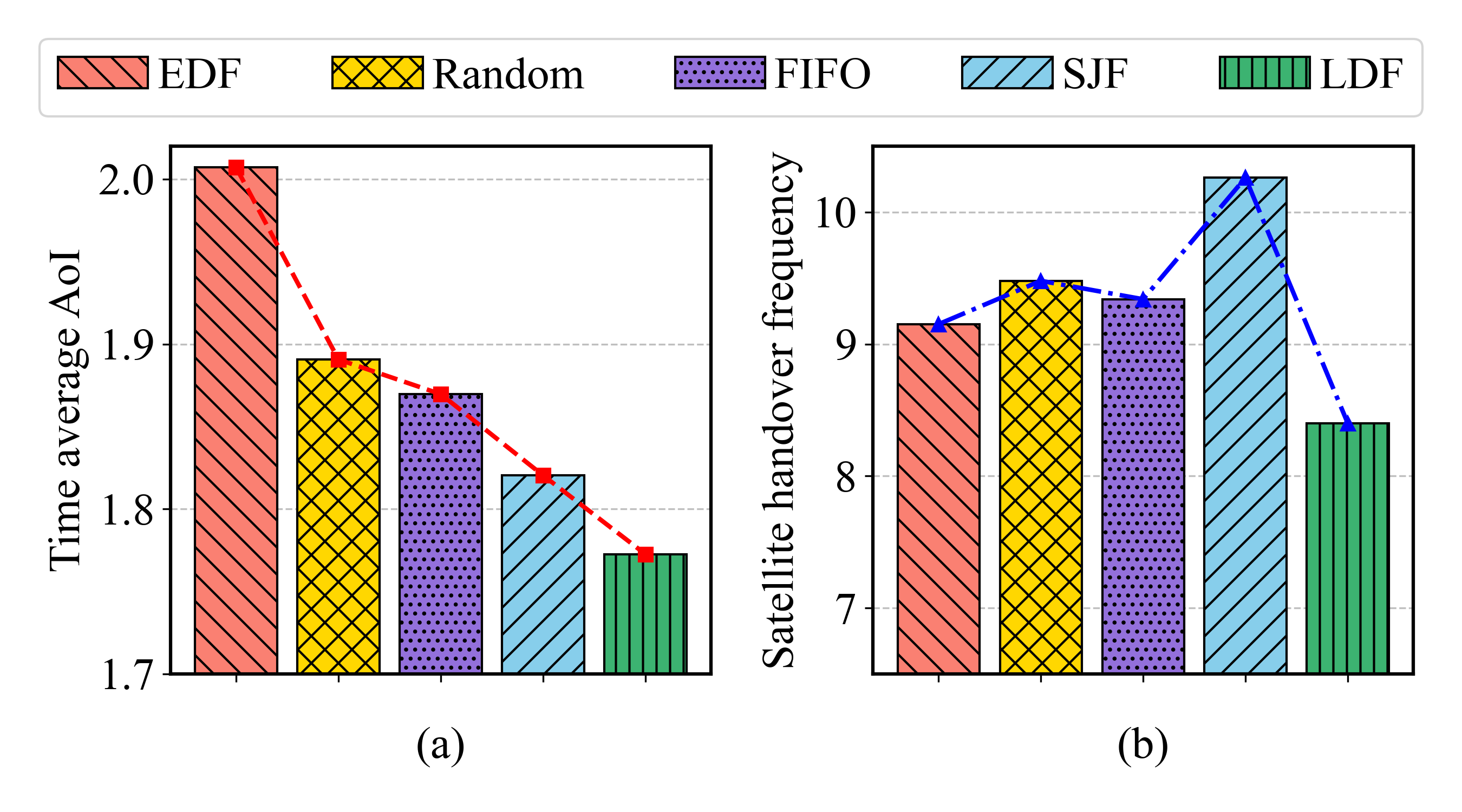}}
    \caption{Performance comparison of different HAP scheduling policies.}
    \label{fig:hap_policy}
\end{figure}

\subsection{Impact of System Settings}
%Comparison of Queue Scheduling Strategies at the HAP Layer
\subsubsection{Comparative Analysis of Different HAP Scheduling Policies}
\par To assess the impact of queue scheduling mechanisms at the HAP layer on overall system performance, we introduce a set of comparative experiments. Specifically, the scheduling policies are as follows: 
%To assess the impact of queue scheduling mechanisms on overall system performance, we introduce a set of comparative experiments by integrating multiple classical queue scheduling policies at the HAP layer. Specifically, the scheduling policies are as follows: 
%In particular, the DRL agent responsible for satellite selection remains unchanged across all experiments. 
%These policies govern the order in which buffered data packets are transmitted from the HAP to ground users.
%By isolating the effect of scheduling policies, we aim to evaluate how local queuing behavior at the HAP affects global performance metrics such as AoI and satellite handover frequency. 

\begin{itemize}
    \item \textit{Random}: A non-deterministic policy where a packet is randomly selected from the buffer at each step. 
    \item \textit{First-In First-Out (FIFO)}: The default policy where packets are served in the order they arrive at the HAP buffer. 
    \item \textit{Earliest Deadline First (EDF)}: Packets with the earliest absolute deadlines are prioritized for transmission. %(\textit{e.g.}, derived from their source generation timestamps) 
    \item \textit{latest deadline first (LDF)}: Packets with the most relaxed deadlines are prioritized. %, which represents a delay-tolerant service preference. 
    \item \textit{Shortest Job First (SJF)}: Packets associated with the smallest estimated transmission time are prioritized. 
\end{itemize}
%\par The simulation results are summarized in Fig.~\ref{fig:hap_policy}, which illustrates the performance of the DRL agent under each scheduling policy in terms of (a) average AoI, and (b) satellite switching frequency, averaged over 50 evaluation episodes. Each bar represents the final converged performance using a fixed trained policy, and the DRL training process was kept consistent across all scheduling configurations. 

\par Fig.~\ref{fig:hap_policy}(a) shows the time average AoI across different scheduling policies. We observe that the LDF policy achieves the lowest average AoI among all compared strategies. The superior performance of LDF can be attributed to its mechanism of prioritizing packets that have been buffered for longer durations, which effectively minimizes information staleness. In contrast, the EDF policy results in significantly higher AoI values due to its fundamental characteristic of delaying the transmission of time-sensitive information. %Moreover, the SJF approach demonstrates moderately favorable AoI performance, which indicates that prioritizing shorter packets contributes to improved information freshness in bandwidth-constrained environments. 
\par Fig.~\ref{fig:hap_policy}(b) presents the satellite handover frequency across different scheduling policies. As can be seen, SJF exhibits the highest handover frequency, whereas LDF maintains the lowest frequency. The frequent handover under SJF is mainly caused by its preference for short packets, which often leads to unstable scheduling decisions across satellites. In contrast, LDF inherently stabilizes the communication process by continuously serving long-waiting packets, thereby reducing the need for frequent satellite changes. 
%\par Fig.~\ref{fig:hap_policy}(b) presents the satellite handover frequency across different scheduling policies. As can be seen, SJF exhibits the highest handover frequency, and LDF maintains the lowest frequency. Moreover, the SJF policy generates substantially more satellite switching events, despite its moderate performance in AoI reduction shown in Fig.~\ref{fig:hap_policy}(a), thereby potentially introducing network instability. %The superior stability of LDF policy can be attributed to its tendency to maintain established communication paths rather than pursuing optimal deadline management. 
\par Overall, these results indicate that the choice of HAP queue scheduling policy significantly influences the performance of the DRL-based satellite selection, with LDF providing the best trade-off between timeliness and handover stability. 
%\par Overall, these results suggest that the choice of queue scheduling strategy at the HAP plays a non-negligible role in shaping the downstream effects of the DRL-based satellite selection policy. In particular, LDF excels in timeliness and strikes a balance between AoI minimization and satellite handover frequency reduction. 

\subsubsection{Impact of Numbers of Ground Users}
\par Fig.~\ref{fig:num_users} illustrates the relationship between the number of ground users and two key performance metrics, which are the time average AoI and satellite handover frequency. As can be seen, the time-average AoI demonstrates a gradual increase as the number of users grows from 1 to 6 and reaches its peak value of approximately 1.82 at 6 users, followed by a slight decline thereafter. Moreover, the satellite handover frequency exhibits a gradual decline as the user count increases from 1 to 6, reaching its minimum at approximately 8.2 handover frequency when serving 6 users. Beyond 6 users, the handover frequency rises again. The above observations primarily result from the downlink bottleneck at the HAP layer, where increasing number of users must compete for fixed transmission resources. 
%\par To verify the impact of the number of ground users on system performance, we performed a detailed simulation under varying numbers of ground users. As shown in Fig.~\ref{fig:num_users}, as the number of ground users increases beyond 6, a noticeable degradation in time-average AoI is observed. This is mainly due to the downlink bottleneck at the HAP layer, where more users share the same transmission capacity. Meanwhile, since the DRL policy only controls satellite selection and is not trained under varying user load, its decisions may lead to suboptimal scheduling when ground users exceed the training configuration. These results suggest that future work should consider incorporating user dynamics into the DRL policy design and training.

\begin{figure}[t]
    \centerline{\includegraphics[width=\linewidth]{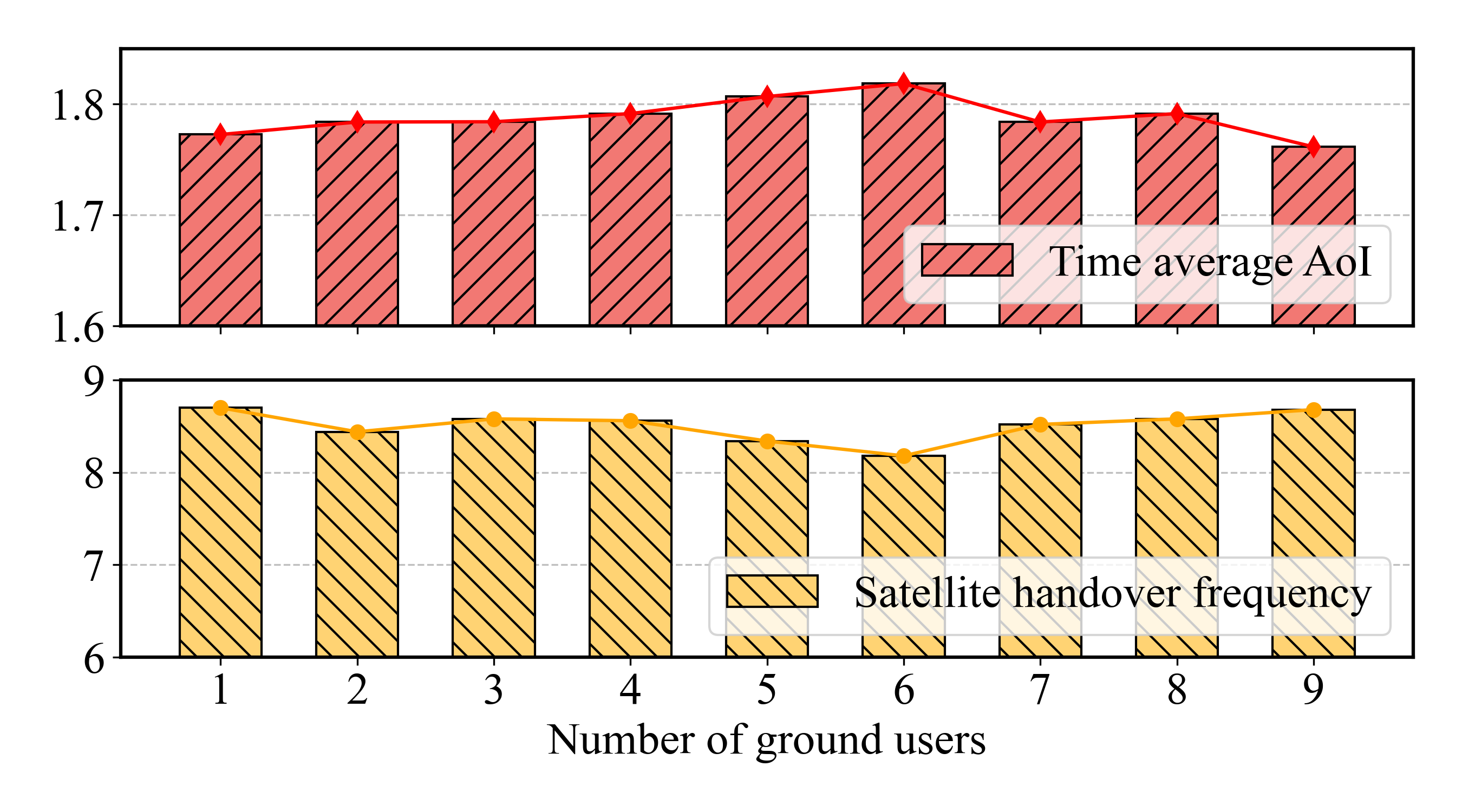}}
    \caption{Comparison of curves of DD3QN-AS with different numbers of ground users.}
    \label{fig:num_users}
\end{figure}

\section{Conclusion}\label{sec:conclusion}
In this paper, we have investigated an AoI-aware HAP-assisted downlink system for SAGIN that employs hybrid FSO/RF communication. We have formulated a joint optimization problem to simultaneously balance the AoI minimization and satellite handover frequency through optimal transmit power distribution and satellite selection decisions. The problem has been proven highly challenging due to its dynamic and non-convex nature with time-coupled constraints that require enhanced adaptability to manage temporal variations in orbital environments. To this end, we have proposed the DD3QN-AS algorithm which effectively processes high-dimensional state information and enables robust decision-making under temporal variations and non-stationary channel conditions. Simulation results have demonstrated that the proposed DD3QN-AS algorithm converges faster and achieves lower AoI and handover frequency than the baselines. Furthermore, we have found that the LDF policy achieves the best balance between AoI minimization and handover stability across several HAP queue scheduling policies. In addition, we have verified the robustness of the DD3QN-AS by extensive performance evaluations across different denoising steps and varying numbers of ground users. 
%In addition, performance analysis under different denoising steps of DD3QN-AS and different numbers of ground users verifies the robustness of the proposed approach. 
Future work will focus on extending the generative artificial intelligence model-based DRL approach to more complex SAGIN scenarios and exploring hybrid learning frameworks that could further enhance system performance in highly dynamic satellite communication environments.

\ifCLASSOPTIONcaptionsoff
  \newpage
\fi

\bibliographystyle{IEEEtran}  % 选择适当的参考文献格式
\bibliography{SAGINs}     % 指向您的 .bib 文件名（如 references.bib）

\begin{thebibliography}{10}
\providecommand{\url}[1]{#1}
\csname url@samestyle\endcsname
\providecommand{\newblock}{\relax}
\providecommand{\bibinfo}[2]{#2}
\providecommand{\BIBentrySTDinterwordspacing}{\spaceskip=0pt\relax}
\providecommand{\BIBentryALTinterwordstretchfactor}{4}
\providecommand{\BIBentryALTinterwordspacing}{\spaceskip=\fontdimen2\font plus
\BIBentryALTinterwordstretchfactor\fontdimen3\font minus \fontdimen4\font\relax}
\providecommand{\BIBforeignlanguage}[2]{{%
\expandafter\ifx\csname l@#1\endcsname\relax
\typeout{** WARNING: IEEEtran.bst: No hyphenation pattern has been}%
\typeout{** loaded for the language `#1'. Using the pattern for}%
\typeout{** the default language instead.}%
\else
\language=\csname l@#1\endcsname
\fi
#2}}
\providecommand{\BIBdecl}{\relax}
\BIBdecl

\bibitem{Wang2025}
J.~Wang, W.~Cheng, and W.~Zhang, ``Performance boundary analyses for statistical multi-{Q}o{S} framework over 6{G} {SAGIN}s,'' \emph{{IEEE} Trans. Wirel. Commun.}, vol.~24, no.~5, pp. 4329--4343, 2025.

\bibitem{Li2024}
J.~Li, G.~Sun, Q.~Wu, D.~Niyato, J.~Kang, A.~Jamalipour, and V.~C.~M. Leung, ``Collaborative ground-space communications via evolutionary multi-objective deep reinforcement learning,'' \emph{{IEEE} J. Sel. Areas Commun.}, vol.~42, no.~12, pp. 3395--3411, 2024.

\bibitem{fan2025satellite}
W.~Fan, Q.~Meng, G.~Wang, H.~Bian, Y.~Liu, and Y.~Liu, ``Satellite edge intelligence: {DRL}-based resource management for task inference in {LEO}-based satellite-ground collaborative networks,'' \emph{{IEEE} Transactions on Mobile Computing}, 2025.

\bibitem{Hui2025}
M.~Hui, S.~Zhai, D.~Wang, T.~Hui, W.~Wang, P.~Du, and F.~Gong, ``A review of {LEO}-satellite communication payloads for integrated communication, navigation, and remote sensing: Opportunities, challenges, future directions,'' \emph{{IEEE} Internet Things J.}, vol.~12, no.~12, pp. 18\,954--18\,992, 2025.

\bibitem{Chen2024a}
J.~Chen, W.~Kuo, and W.~Liao, ``Space{E}dge: Optimizing service latency and sustainability for space-centric task offloading in {LEO} satellite networks,'' \emph{{IEEE} Trans. Wirel. Commun.}, vol.~23, no.~10, pp. 15\,435--15\,446, 2024.

\bibitem{Chen2025}
Y.~Chen, J.~Zhao, Y.~Wu, J.~Huang, and X.~S. Shen, ``Multi-user task offloading in {UAV}-assisted {LEO} satellite edge computing: {A} game-theoretic approach,'' \emph{{IEEE} Trans. Mob. Comput.}, vol.~24, no.~1, pp. 363--378, 2025.

\bibitem{zhang2025aoi}
G.~Zhang, X.~Wei, X.~Tan, Z.~Han, and G.~Zhang, ``{A}o{I} minimization based on deep reinforcement learning and matching game for {I}o{T} information collection in {SAGIN},'' \emph{{IEEE} Transactions on Communications}, 2025.

\bibitem{li2025llm}
J.~Li, G.~Sun, Z.~Sun, J.~Wang, Y.~Liu, R.~Zhang, D.~Niyato, and S.~Mao, ``{LLM}-guided {DRL} for multi-tier {LEO} satellite networks with hybrid {FSO/RF} links,'' \emph{arXiv preprint arXiv:2505.11978}, 2025.

\bibitem{Le2023}
H.~D. Le, H.~D. Nguyen, C.~T. Nguyen, and A.~T. Pham, ``{FSO}-based space-air-ground integrated vehicular networks: Cooperative {HARQ} with rate adaptation,'' \emph{{IEEE} Trans. Aerosp. Electron. Syst.}, vol.~59, no.~4, pp. 4076--4091, 2023.

\bibitem{Mao2024}
B.~Mao, X.~Zhou, J.~Liu, and N.~Kato, ``On an intelligent hierarchical routing strategy for ultra-dense free space optical low earth orbit satellite networks,'' \emph{{IEEE} J. Sel. Areas Commun.}, vol.~42, no.~5, pp. 1219--1230, 2024.

\bibitem{Dabiri2025}
M.~T. Dabiri, M.~O. Hasna, S.~Althunibat, and K.~A. Qaraqe, ``Modulating retroreflector-based satellite-to-ground optical links: Joint communications and tracking,'' \emph{{IEEE} Trans. Commun.}, vol.~73, no.~3, pp. 1950--1962, 2025.

\bibitem{chen2025free}
W.~Chen, C.~Ju, T.~Yuan, Y.~Zhan, M.~Zhang, and D.~Wang, ``Free space optical semantic communication for satellite remote sensing image transmission,'' \emph{{IEEE} Transactions on Communications}, 2025.

\bibitem{Ma2022}
Y.~Ma, T.~Lv, G.~Pan, Y.~Chen, and M.~Alouini, ``On secure uplink transmission in hybrid {RF-FSO} cooperative satellite-aerial-terrestrial networks,'' \emph{{IEEE} Trans. Commun.}, vol.~70, no.~12, pp. 8244--8257, 2022.

\bibitem{xu2024cooperative}
G.~Xu, M.~Xu, Q.~Zhang, and Z.~Song, ``Cooperative {FSO/RF} space-air-ground integrated network system with adaptive combining: A performance analysis,'' \emph{{IEEE} Transactions on Wireless Communications}, 2024.

\bibitem{mashiko2025combined}
K.~Mashiko, Y.~Kawamoto, N.~Kato, K.~Yoshida, and M.~Ariyoshi, ``Combined control of coverage area and {HAP}s deployment in hybrid {FSO/RF} {SAGIN},'' \emph{{IEEE} Transactions on Vehicular Technology}, 2025.

\bibitem{fan2025gato}
Y.~Fan, Y.~Bi, Y.~Liu, D.~Niyato, L.~Zhao, Q.~He, and A.~Hawbani, ``{GATO}: Global transmission optimization for {SAGIN}-assisted {I}o{RT} data collection,'' \emph{{IEEE} Transactions on Mobile Computing}, 2025.

\bibitem{Gao2024}
Y.~Gao, Z.~Ye, and H.~Yu, ``Cost-efficient computation offloading in {SAGIN:} {A} deep reinforcement learning and perception-aided approach,'' \emph{{IEEE} J. Sel. Areas Commun.}, vol.~42, no.~12, pp. 3462--3476, 2024.

\bibitem{tan2024outage}
J.~Tan, F.~Tang, M.~Zhao, and N.~Kato, ``Outage probability, performance, fairness analysis of space-air-ground integrated network ({SAGIN}): {UAV} altitude and position angle,'' \emph{{IEEE} Transactions on Wireless Communications}, 2024.

\bibitem{zhao2025energy}
Z.~Zhao, Z.~Yang, M.~Chen, C.~Zhu, W.~Xu, Z.~Zhang, and K.~Huang, ``Energy-efficient probabilistic semantic communication over space-air-ground integrated networks,'' \emph{{IEEE} Transactions on Wireless Communications}, 2025.

\bibitem{zhang2025integrated}
Y.~Zhang, Y.~Gang, P.~Wu, G.~Fan, W.~Xu, B.~Ai, and Q.~Wu, ``Integrated sensing, communication, and computation in {SAGIN}: Joint beamforming and resource allocation,'' \emph{{IEEE} Transactions on Cognitive Communications and Networking}, 2025.

\bibitem{Ke2025}
Y.~Ke, Z.~Ni, D.~Zhang, X.~Miao, C.~Y. Leow, S.~Wang, G.~Pan, and J.~An, ``Information freshness in multi-hop satellite {I}o{T} systems,'' \emph{{IEEE} Trans. Mob. Comput.}, vol.~24, no.~7, pp. 6014--6029, 2025.

\bibitem{Deng2024}
Y.~Deng, S.~Wu, J.~You, Y.~Wang, X.~Zhang, and Q.~Zhang, ``Age-energy tradeoff of polar-coded {HARQ-CC} in space-air-ground integrated network,'' \emph{{IEEE} Trans. Veh. Technol.}, vol.~73, no.~7, pp. 9943--9957, 2024.

\bibitem{Wang2024}
Y.~Wang, S.~Wu, Y.~Wang, X.~Zhang, J.~Jiao, and Q.~Zhang, ``Goal-oriented transmission scheduling for energy-efficient wireless networked control in {SAGIN:} an {A}o{I}-thresholding mechanism,'' \emph{{IEEE} Trans. Veh. Technol.}, vol.~73, no.~8, pp. 11\,503--11\,517, 2024.

\bibitem{Nguyen2024}
M.~D. Nguyen, L.~B. Le, and A.~Girard, ``Integrated computation offloading, {UAV} trajectory control, edge-cloud and radio resource allocation in {SAGIN},'' \emph{{IEEE} Trans. Cloud Comput.}, vol.~12, no.~1, pp. 100--115, 2024.

\bibitem{Bao2024}
J.~Bao, P.~Wang, L.~Zhao, C.~Liu, and B.~Jiang, ``Information-centric wireless sensor {SAGIN} with decentralized caching status aware multiple subsystem nested coded caching,'' \emph{{IEEE} Internet Things J.}, vol.~11, no.~23, pp. 37\,899--37\,915, 2024.

\bibitem{Sun2024}
G.~Sun, L.~He, Z.~Sun, Q.~Wu, S.~Liang, J.~Li, D.~Niyato, and V.~C.~M. Leung, ``Joint task offloading and resource allocation in aerial-terrestrial {UAV} networks with edge and fog computing for post-disaster rescue,'' \emph{{IEEE} Trans. Mob. Comput.}, vol.~23, no.~9, pp. 8582--8600, 2024.

\bibitem{Mao2025}
W.~Mao, Y.~Lu, G.~Pan, and B.~Ai, ``{UAV}-assisted communications in {SAGIN-ISAC:} mobile user tracking and robust beamforming,'' \emph{{IEEE} J. Sel. Areas Commun.}, vol.~43, no.~1, pp. 186--200, 2025.

\bibitem{Elmahallawy2024}
M.~Elmahallawy, T.~Luo, and K.~Ramadan, ``Communication-efficient federated learning for {LEO} constellations integrated with {HAP}s using hybrid {NOMA-OFDM},'' \emph{{IEEE} J. Sel. Areas Commun.}, vol.~42, no.~5, pp. 1097--1114, 2024.

\bibitem{Qin2024}
Y.~Qin, Y.~Yang, F.~Tang, X.~Yao, M.~Zhao, and N.~Kato, ``Differentiated federated reinforcement learning based traffic offloading on space-air-ground integrated networks,'' \emph{{IEEE} Trans. Mob. Comput.}, vol.~23, no.~12, pp. 11\,000--11\,013, 2024.

\bibitem{Chen2024}
Z.~Chen, J.~Zhang, G.~Min, Z.~Ning, and J.~Li, ``Traffic-aware lightweight hierarchical offloading toward adaptive slicing-enabled {SAGIN},'' \emph{{IEEE} J. Sel. Areas Commun.}, vol.~42, no.~12, pp. 3536--3550, 2024.

\bibitem{Sun2024a}
G.~Sun, Y.~Wang, H.~Yu, and M.~Guizani, ``Proportional fairness-aware task scheduling in space-air-ground integrated networks,'' \emph{{IEEE} Trans. Serv. Comput.}, vol.~17, no.~6, pp. 4125--4137, 2024.

\bibitem{Xiao2024}
Z.~Xiao, J.~Yang, T.~Mao, C.~Xu, R.~Zhang, Z.~Han, and X.~Xia, ``{LEO} satellite access network {(LEO-SAN)} toward 6{G}: Challenges and approaches,'' \emph{{IEEE} Wirel. Commun.}, vol.~31, no.~2, pp. 89--96, 2024.

\bibitem{wu2024deep}
M.~Wu, K.~Guo, X.~Li, Z.~Lin, Y.~Wu, T.~A. Tsiftsis, and H.~Song, ``Deep reinforcement learning-based energy efficiency optimization for {RIS}-aided integrated satellite-aerial-terrestrial relay networks,'' \emph{{IEEE} Transactions on Communications}, vol.~72, no.~7, pp. 4163--4178, 2024.

\bibitem{Bian2025}
C.~Bian, Y.~Shao, H.~Wu, E.~Ozfatura, and D.~G{\"{u}}nd{\"{u}}z, ``Process-and-forward: Deep joint source-channel coding over cooperative relay networks,'' \emph{{IEEE} J. Sel. Areas Commun.}, vol.~43, no.~4, pp. 1118--1134, 2025.

\bibitem{Abderrahim2024}
W.~Abderrahim, O.~Amin, and B.~Shihada, ``Data center-enabled high altitude platforms: {A} green computing alternative,'' \emph{{IEEE} Trans. Mob. Comput.}, vol.~23, no.~5, pp. 6149--6162, 2024.

\bibitem{Belotti2013}
P.~Belotti, C.~Kirches, S.~Leyffer, J.~T. Linderoth, J.~R. Luedtke, and A.~Mahajan, ``Mixed-integer nonlinear optimization,'' \emph{Acta Numer.}, vol.~22, pp. 1--131, 2013.

\bibitem{Wang2016}
Z.~Wang, T.~Schaul, M.~Hessel, H.~van Hasselt, M.~Lanctot, and N.~de~Freitas, ``Dueling network architectures for deep reinforcement learning,'' in \emph{Proc. {ICML} 2016}, vol.~48, 2016, pp. 1995--2003.

\bibitem{Cai2023}
Q.~Cai, C.~Cui, Y.~Xiong, W.~Wang, Z.~Xie, and M.~Zhang, ``A survey on deep reinforcement learning for data processing and analytics,'' \emph{{IEEE} Trans. Knowl. Data Eng.}, vol.~35, no.~5, pp. 4446--4465, 2023.

\bibitem{ArjonaMedina2019}
J.~A. Arjona{-}Medina, M.~Gillhofer, M.~Widrich, T.~Unterthiner, J.~Brandstetter, and S.~Hochreiter, ``{RUDDER:} return decomposition for delayed rewards,'' in \emph{Proc. {NeurIPS} 2019}, pp. 13\,544--13\,555.

\bibitem{Chi2024}
J.~Chi, X.~Zhou, F.~Xiao, Y.~Lim, and T.~Qiu, ``Task offloading via prioritized experience-based double dueling {DQN} in edge-assisted {II}o{T},'' \emph{{IEEE} Trans. Mob. Comput.}, vol.~23, no.~12, pp. 14\,575--14\,591, 2024.

\bibitem{Fan2019}
Z.~Fan, R.~Su, W.~Zhang, and Y.~Yu, ``Hybrid actor-critic reinforcement learning in parameterized action space,'' in \emph{Proc. {IJCAI} 2019}, pp. 2279--2285.

\bibitem{Chen2021}
L.~Chen, K.~Lu, A.~Rajeswaran, K.~Lee, A.~Grover, M.~Laskin, P.~Abbeel, A.~Srinivas, and I.~Mordatch, ``Decision transformer: Reinforcement learning via sequence modeling,'' in \emph{Proc. {N}eur{IPS} 2021}, pp. 15\,084--15\,097.

\bibitem{Rombach2022}
R.~Rombach, A.~Blattmann, D.~Lorenz, P.~Esser, and B.~Ommer, ``High-resolution image synthesis with latent diffusion models,'' in \emph{Proc. {IEEE} {CVPR} 2022}, 2022, pp. 10\,674--10\,685.

\bibitem{Shribak2024}
D.~Shribak, C.~Gao, Y.~Li, C.~Xiao, and B.~Dai, ``Diffusion spectral representation for reinforcement learning,'' in \emph{Proc. {N}eur{IPS} 2024}, 2024.

\bibitem{Liu2025a}
Z.~Liu, H.~Du, J.~Lin, Z.~Gao, L.~Huang, S.~Hosseinalipour, and D.~Niyato, ``{DNN} partitioning, task offloading, and resource allocation in dynamic vehicular networks: {A} {L}yapunov-guided diffusion-based reinforcement learning approach,'' \emph{{IEEE} Trans. Mob. Comput.}, vol.~24, no.~3, pp. 1945--1962, 2025.

\bibitem{Kingma2015}
D.~P. Kingma and J.~Ba, ``Adam: {A} method for stochastic optimization,'' in \emph{{Proc.} {ICLR} 2015}, 2015.

\bibitem{Haarnoja2018}
T.~Haarnoja, A.~Zhou, K.~Hartikainen, G.~Tucker, S.~Ha, J.~Tan, V.~Kumar, H.~Zhu, A.~Gupta, P.~Abbeel, and S.~Levine, ``Soft actor-critic algorithms and applications,'' \emph{Co{RR}}, vol. abs/1812.05905, 2018.

\bibitem{Kuznetsov2020}
A.~Kuznetsov, P.~Shvechikov, A.~Grishin, and D.~P. Vetrov, ``Controlling overestimation bias with truncated mixture of continuous distributional quantile critics,'' in \emph{Proc. {ICML} 2020}, 2020, pp. 5556--5566.

\bibitem{Mnih2015}
V.~Mnih, K.~Kavukcuoglu, D.~Silver, A.~A. Rusu, J.~Veness, M.~G. Bellemare, A.~Graves, M.~A. Riedmiller, A.~Fidjeland, G.~Ostrovski, S.~Petersen, C.~Beattie, A.~Sadik, I.~Antonoglou, H.~King, D.~Kumaran, D.~Wierstra, S.~Legg, and D.~Hassabis, ``Human-level control through deep reinforcement learning,'' \emph{Nature}, vol. 518, no. 7540, pp. 529--533, 2015.

\bibitem{Hasselt2016}
H.~van Hasselt, A.~Guez, and D.~Silver, ``Deep reinforcement learning with double {Q}-learning,'' in \emph{Proc. {AAAI} 2016}, 2016, pp. 2094--2100.

\bibitem{Zhang2025}
C.~Zhang, G.~Sun, J.~Li, Q.~Wu, J.~Wang, D.~Niyato, and Y.~Liu, ``Multi-objective aerial collaborative secure communication optimization via generative diffusion model-enabled deep reinforcement learning,'' \emph{{IEEE} Trans. Mob. Comput.}, vol.~24, no.~4, pp. 3041--3058, 2025.

\end{thebibliography}
\vspace{-1.2cm}
\begin{IEEEbiography}[{\includegraphics[width=1in,height=1.25in,clip,keepaspectratio]{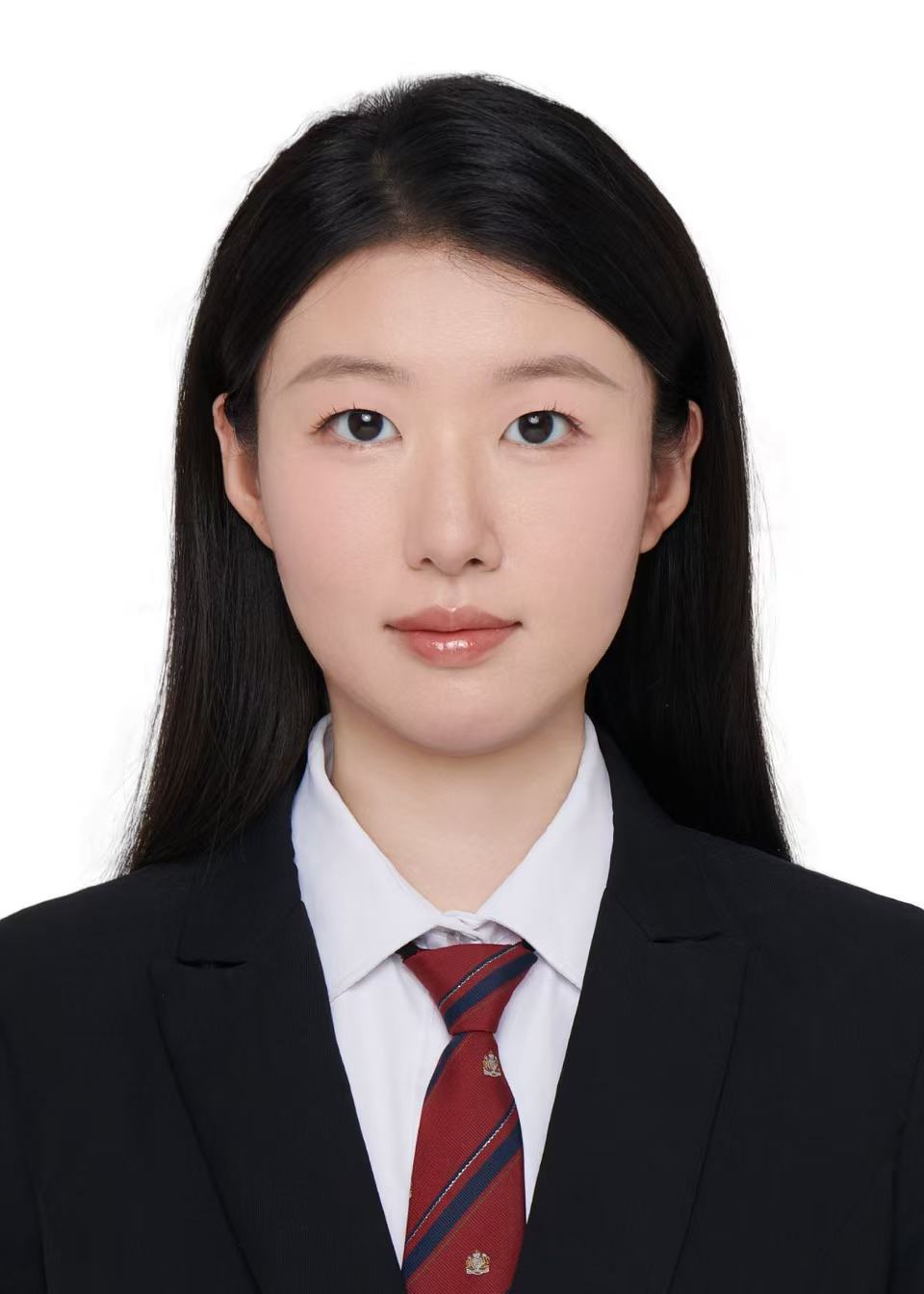}}]{Zifan Lang} received the B.S. degree in College of Computer Science and Technology, Jilin University, Changchun, China, in 2023. She is currently working toward the Ph.D. degree in computer technology with Jilin University, Changchun, China. Her research interests include UAV communications, space-air-ground integrated networks, and deep reinforcement learning. 
\end{IEEEbiography}
\vspace{-1.1cm}
\begin{IEEEbiography}[{\includegraphics[width=1in,height=1.25in,clip,keepaspectratio]{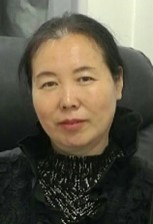}}]{Guixia Liu} received the B.S, M.S, and Ph.D. degree in College of Computer Science and Technology, Jilin University, Changchun, China. She is a research professor in the College of Computer Science and Technology, Jilin University, Changchun, China. Her research interests include machine learning, deep learning, medical image processing and analysis, and biological big data mining and analysis. She has published several papers in journals, BIB, PR and AI.
\end{IEEEbiography}
\vspace{-1.1cm}
\begin{IEEEbiography}[{\includegraphics[width=1in,height=1.25in,clip,keepaspectratio]{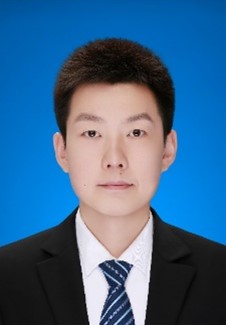}}]{Geng Sun} (Senior Member, IEEE) received the B.S. degree in communication engineering from Dalian Polytechnic University, and the Ph.D. degree in computer science and technology from Jilin University, in 2011 and 2018, respectively. He was a Visiting Researcher with the School of Electrical and Computer Engineering, Georgia Institute of Technology, USA. He is a Professor in the College of Computer Science and Technology at Jilin University. Currently, he is working as a visiting scholar at the College of Computing and Data Science, Nanyang Technological University, Singapore. He has published over 100 high-quality papers, including IEEE TMC, IEEE JSAC, IEEE/ACM ToN, IEEE TWC, IEEE TCOM, IEEE TAP, IEEE IoT-J, IEEE TIM, IEEE INFOCOM, IEEE GLOBECOM, and IEEE ICC. He serves as the Associate Editors of IEEE Communications Surveys \& Tutorials, IEEE Transactions on Vehicular Technology, IEEE Transactions on Network Science and Engineering, and IEEE Networking Letters. He serves as the Lead Guest Editor of Special Issues for IEEE Transactions on Network Science and Engineering, IEEE Internet of Things Journal, IEEE Networking Letters. He also serves as the Guest Editor of Special Issues for IEEE Transactions on Services Computing, IEEE Communications Magazine, and IEEE Open Journal of the Communications Society. His research interests include UAV communications and networking, mobile edge computing (MEC), intelligent reflecting surface (IRS), generative AI, and deep reinforcement learning.
\end{IEEEbiography}

\begin{IEEEbiography}[{\includegraphics[width=1in,height=1.25in,clip,keepaspectratio]{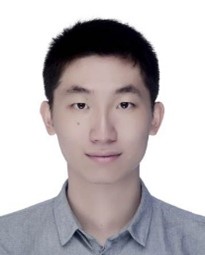}}]{Jiahui Li} (Member, IEEE) received his B.S. in Software Engineering, and M.S. and Ph.D. in Computer Science and Technology from Jilin University, Changchun, China, in 2018, 2021, and 2024, respectively. He was a visiting Ph.D. student at the Singapore University of Technology and Design (SUTD). He currently serves as an assistant researcher in the College of Computer Science and Technology at Jilin University. His current research focuses on integrated air-ground networks, UAV networks, wireless energy transfer, and optimization.
\end{IEEEbiography}

\begin{IEEEbiography}[{\includegraphics[width=1in,height=1.25in,clip,keepaspectratio]{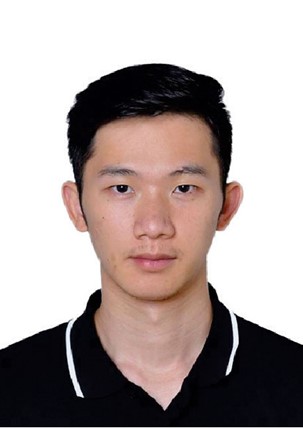}}]{Jiacheng Wang} received the Ph.D. degree from the School of Communication and Information Engineering, Chongqing University of Posts and Telecommunications, Chongqing, China. He is currently a Research Associate in computer science and engineering with Nanyang Technological University, Singapore. His research interests include wireless sensing, semantic communications, and metaverse. 
\end{IEEEbiography}

\begin{IEEEbiography}[{\includegraphics[width=1in,height=1.25in,clip,keepaspectratio]{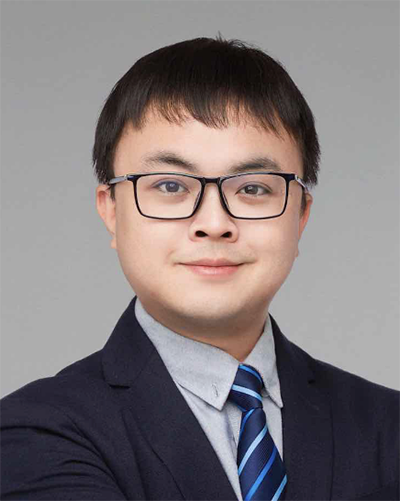}}]{Weijie Yuan} (Senior Member, IEEE) received the Ph.D. degree from the University of Technology Sydney, Australia, in 2019. In 2016, he was a Visiting Ph.D. Student with the Institute of Telecommunications, Vienna University of Technology, Austria. He was a Research Assistant with The University of Sydney, a Visiting Associate Fellow with the University of Wollongong, and a Visiting Fellow with the University of Southampton, from 2017 to 2019. From 2019 to 2021, he was a Research Associate with the University of New South Wales. He is currently an Assistant Professor with Southern University of Science and Technology. He currently serves as an Associate Editor for IEEE TRANSACTIONS ON WIRELESS COMMUNICATIONS, IEEE Communications Standards Magazine, IEEE TRANSACTIONS ON GREEN COMMUNICATIONS AND NETWORKING, IEEE COMMUNICATIONS LETTERS, IEEE OPEN JOURNAL OF COMMUNICATIONS SOCIETY, and EURASIP Journal on Advances in Signal Processing. He is the Lead Editor for the series on ISAC in IEEE Communications Magazine. He was an Organizer/the Chair of several workshops and special sessions on OTFS and ISAC in flagship IEEE and ACM conferences, including IEEE ICC, IEEE/CIC ICCC, IEEE SPAWC, IEEE VTC, IEEE WCNC, IEEE ICASSP, and ACM MobiCom. He is the Founding Chair of the IEEE ComSoc Special Interest Group on OTFS (OTFS-SIG). He has been listed in the World’s Top 2$\%$ Scientists by Stanford University for citation impact since 2021. He was a recipient of the Best Editor Award from IEEE COMMUNICATIONS LETTERS, the Exemplary Reviewer from IEEE TRANSACTIONS ON COMMUNICATIONS/IEEE WIRELESS COMMUNICATIONS LETTERS, and the Best Paper Award from IEEE ICC 2023, IEEE/CIC ICCC 2023, and IEEE GlobeCom 2024.
\end{IEEEbiography}

\begin{IEEEbiography}[{\includegraphics[width=1in,height=1.25in,clip,keepaspectratio]{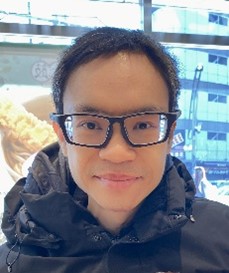}}]{Dusit Niyato} (Fellow, IEEE) received the B.Eng. degree from the King Mongkuts Institute of Technology Ladkrabang (KMITL), Thailand, in 1999, and the Ph.D. degree in electrical and computer engineering from the University of Manitoba, Canada, in 2008. He is currently a Professor with the School of Computer Science and Engineering, Nanyang Technological University, Singapore. His research interests include the Internet of Things (IoT), machine learning, and incentive mechanism design. 
\end{IEEEbiography}

\begin{IEEEbiography}[{\includegraphics[width=1in,height=1.25in,clip,keepaspectratio]{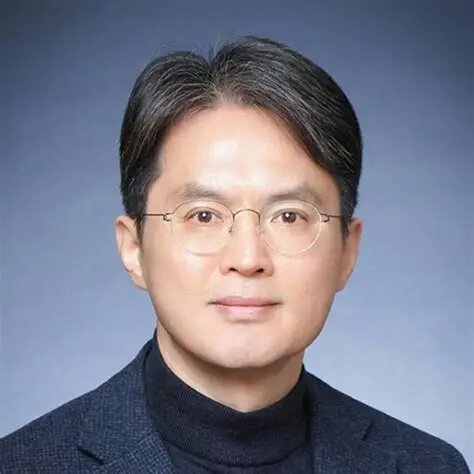}}]{Dong In Kim} (Life Fellow, IEEE) received the Ph.D. degree in electrical engineering from the University of Southern California at Los Angeles, CA, USA, in 1990. He was a tenured Professor with the School of Engineering Science, Simon Fraser University, Burnaby, BC, Canada. He is currently a Distinguished Professor with the College of Information and Communication Engineering, Sungkyunkwan University, Suwon, South Korea. He is a fellow of Korean Academy of Science and Technology and a Life Member of the National Academy of Engineering of Korea. He was the first recipient of the NRF of Korea Engineering Research Center in Wireless Communications for RF Energy Harvesting from 2014 to 2021. He received several research awards, including the 2023 IEEE ComSoc Best Survey Paper Award and the 2022 IEEE Best Land Transportation Paper Award. He was selected the 2019 recipient of the IEEE ComSoc Joseph LoCicero Award for Exemplary Service to Publications. He was the General Chair of the IEEE ICC 2022, Seoul. He has been listed as a 2020/2022 Highly Cited Researcher by Clarivate Analytics. From 2001 to 2024, he served as an Editor, an Editor-at-Large, and an Area Editor of Wireless Communications I for IEEE Transactions on Communications. From 2002 to 2011, he served as an Editor and a Founding Area Editor of Cross-Layer Design and Optimization for IEEE Transactions on Wireless Communications. From 2008 to 2011, he served as the Co-Editor-in-Chief for the IEEE/KICS Journal of Communications and Networks. He served as the Founding Editor-in-Chief for IEEE Wireless Communications Letters from 2012 to 2015. 
\end{IEEEbiography}

%\newpage

\end{document}